%% file: main.tex
\documentclass[conference,a4paper]{IEEEtran}

\input{packages}

\input{commands.tex}

\IEEEoverridecommandlockouts

\title{Sparse and Private Distributed Matrix Multiplication with Straggler Tolerance}

\author{%
  \IEEEauthorblockN{
                    Maximilian Egger,
                    Marvin Xhemrishi,
                    Antonia Wachter-Zeh
                    and Rawad Bitar
                    }
                    
  \IEEEauthorblockA{%
                     Institute for Communications Engineering, Technical University of Munich, Munich, Germany \\
                    \{maximilian.egger, marvin.xhemrishi, antonia.wachter-zeh, rawad.bitar\}@tum.de
                    }\\[-1cm]
                                    
\thanks{This project has received funding from the German Research Foundation (DFG) under Grant Agreement Nos. BI 2492/1-1 and WA 3907/7-1.}
}
\date{\today}

\begin{document}

\maketitle
\begin{abstract}
This paper considers the problem of outsourcing the multiplication of two private and sparse matrices to untrusted workers. Secret sharing schemes can be used to tolerate stragglers and guarantee information-theoretic privacy of the matrices. However, traditional secret sharing schemes destroy all sparsity in the offloaded computational tasks. Since exploiting the sparse nature of matrices was shown to speed up the multiplication process, preserving the sparsity of the input matrices in the computational tasks sent to the workers is desirable. It was recently shown that sparsity can be guaranteed at the expense of a weaker privacy guarantee. Sparse secret sharing schemes with only two output shares were constructed. In this work, we construct sparse secret sharing schemes that generalize Shamir's secret sharing schemes for a fixed threshold $t=2$ and an arbitrarily large number of shares. We design our schemes to provide the strongest privacy guarantee given a desired sparsity of the shares under some mild assumptions. We show that increasing the number of shares, i.e., increasing straggler tolerance, incurs a degradation of the privacy guarantee. However, this degradation is negligible when the number of shares is comparably small to the cardinality of the input alphabet.

\end{abstract}
\vspace{-1ex}
\section{Introduction}
\label{sec:intro}
We consider the problem of offloading the multiplication of two private matrices to a set of untrusted workers. 
Matrix multiplication is considered a problem of independent interest as it forms the main computation of several machine learning algorithms such as linear regression and support vector machines. Offloading the multiplication to worker machines suffers from the presence of slow or unresponsive workers referred to as \emph{stragglers} \cite{Deanetal}. Maintaining privacy of the input matrices is of paramount importance since they can contain sensitive data that should not be revealed to untrusted parties. The literature on this topic is very rich. For brevity, we explain the works closely related to ours and refer the interested reader to \cite{ng2021comprehensive,ulukus2021survey,li2020coded} and references within for comprehensive surveys on the topic.

The use of MDS codes for straggler tolerance was introduced for distributed matrix multiplication in \cite{speeding_up_using_codes} and for general gradient descent based machine learning algorithms in \cite{tandonGradientCoding2017}. To ensure privacy of one or both input matrices and tolerate stragglers, coding theoretic techniques based on secret sharing are applied in the literature. One of the main figures of merit of the introduced schemes is the communication cost incurred by the scheme to ensure privacy and straggler tolerance. For instance, fundamental bounds on the communication cost for a fixed straggler tolerance when privacy of one or both input matrices is required were derived in~\cite{chang2018capacity}. Communication-efficient secret sharing is used in \cite{bitar2020minimizing} to mitigate a flexible number of stragglers and ensure an optimal communication cost while guaranteeing the privacy of one of the input matrices. A general framework with optimal recovery threshold for polynomial computations based on Lagrange polynomials was introduced in \cite{LCC}. In \cite{GASP,aliasgari2020private}, the encoding was carefully designed to reduce communication costs by reducing the number of workers required to outsource the computation. In \cite{makkonen2022secure}, a method was proposed to further reduce the communication when the computation of interest is the Gram of one private input matrix.

In this work, we focus on the setting where the input matrices are inherently \emph{sparse}, i.e., they have a relatively high number of entries equal to zero, e.g., \cite{wright2008robust,madarkar2021sparse}. Due to their structure, sparse matrices can be stored, transmitted and/or multiplied efficiently by leveraging the high number of zero entries~\cite{adaptive_sparse_matrix,implementing_sparse_matrix_vector}. Therefore, preserving the sparsity of the matrices is beneficial even when the computation is distributed. Directly applying MDS codes to generate computational tasks offloaded to the workers to tolerate stragglers, such as in \cite{speeding_up_using_codes}, results in sending dense computational tasks to the workers and may incur artificial delays \cite{coded_sparse_mm}. To that end, codes that output sparse computational tasks were investigated in the literature and were shown to outperform the direct use of MDS codes as a black box \cite{coded_sparse_mm,Lagrange_sparse_coding,coded_sparse_matrix_leverages_partial_stragglers}.

However, when it comes to information-theoretic privacy of the input matrices, applying ideas from secret sharing as done in most of the literature destroys all hope of maintaining sparsity. On a high level, to generate the tasks sent to the workers, the input matrices must be encoded (mixed) with random matrices whose entries are drawn independently and uniformly at random from the alphabet to which the input matrices belong. Hence, the tasks sent to the workers look like random matrices, which as a result are dense. 

The problem of generating sparse computational tasks at the expense of relaxing the privacy guarantee is considered in \cite{sparse_ISIT,xhemrishi2022efficient} when the privacy of one input matrix is to be guaranteed. The crux of \cite{sparse_ISIT} is a sparse secret sharing scheme that takes a sparse matrix $\bfA$ as input and outputs two matrices $\bfR$ and $\bfA+\bfR$ that are sparse but not completely independent from $\bfA$, hence violating the perfect information-theoretic privacy. To tolerate stragglers, the input matrix $\bfA$ is divided row-wise into several blocks, each encoded with the proposed encoding scheme. The resulting matrices are distributed to the workers using fractional repetition codes. In \cite{xhemrishi2022efficient}, a fundamental tradeoff between sparsity and privacy was derived and an improved coding scheme was constructed that outputs only two matrices and achieves this tradeoff.

In this work, we consider secret sharing schemes that take as input a sparse matrix $\bfA$ and output a pre-specified number of matrices, $n\geq 2$, of the form $\bfA+x\bfR$ for $x\in \{1,\dots,n\}$. The output matrices are sparse and slightly dependent on $\bfA$. Similarly to \cite{xhemrishi2022efficient}, we show the existence of a fundamental tradeoff between sparsity, privacy and the value of $n$. Namely, we show that for a fixed $n$, increasing the desired sparsity of the output matrices increases their dependency on $\bfA$, and hence incurs a weaker privacy guarantee. Further, increasing $n$ worsens the tradeoff, i.e., for a fixed desired sparsity increasing $n$ implies a weaker privacy guarantee. We construct optimal such secret sharing schemes. Using our construction, one can encode two private input matrices $\bfA$ and $\bfB$ into private and sparse computational tasks sent to the workers. Straggler tolerance is hence a natural property of the scheme and fractional repetition codes would not be needed.

\section{Preliminaries} 
\label{sec:prelim}
\emph{Notation:} %
We denote %
matrices by %
upper-case bold letters, e.g., %
$\bfA$. %
For a given matrix $\bfA$, the random variable representing one of its entries is denoted by the respective uppercase \emph{typewriter} letter, e.g., $\rvA$. A finite field of cardinality $q$ is denoted by $\F_q$ and its multiplicative group by $\F_q^* \define \F_q\setminus \{0\}$. We use calligraphic letters to denote sets, e.g., $\mathcal{Z}$. The set of all positive integers less than or equal to an integer $n$ is denoted by $[n]\define \{1,2,\dots, n\}$. %
The probability mass function (PMF) of a random variable $\rvA$ defined over $\mathbb{F}_q$ is denoted by $\textrm{P}_\rvA = \left[p_1, p_2, \dots, p_q \right]$, i.e., $\Pr(\rvA = i) = p_i$ for all $i\in \mathbb{F}_q$. The $q$-ary entropy of a random variable $\rvA$ is defined by $\textrm{H}_q(\rvA)$. %
Given two random variables $\rvA$ and $\rvB$, their $q$-ary mutual information is denoted by $\I_q(\rvA; \rvB)$, where all logarithms are taken to the base $q$. The Kullback-Leibler (KL) divergence between the PMFs of $\rvA$ and $\rvB$ is denoted by $\kl{\textrm{P}_{\rvA}}{\textrm{P}_{\rvB}}$.%

\emph{System Model:} We consider a setting in which a central node, called \emph{\master}, owns two large matrices $\bfA \in \F_q^{k \times \ell}$ and $\bfB \in \F_q^{\ell \times m}$ that are a realization of sparse random matrices whose entries are independently and identically distributed\footnote{In cases where matrices have correlated entries, our assumptions are motivated by applying existing sparsity-preserving pre-processing techniques, e.g., shuffling, that break the correlation between the entries.}. We denote by $\rvA$ and $\rvB$ the random variables corresponding to a single entry in $\bfA$ and $\bfB$, respectively. The sparsity levels of the matrices are denoted by $\slevel{\bfA}$ and $\slevel{\bfB}$ as defined next.

{\definition(Sparsity level of a matrix)\label{def:sparsity} {The sparsity level $\sparsity(\bfA)$ of a matrix $\bfA$ with entries independently and identically distributed is equal to the probability of an entry $\rvA$ being equal to $0$, i.e.,%
\vspace{-0.2cm}
\begin{equation*}
\slevel{\bfA} \define \sparsity(\bfA) = \Pr\{\rvA = 0\}\,.\vspace{-0.1cm}
\end{equation*}}}

\newcommand{\powersA}{\ensuremath{\beta}}
\newcommand{\powersB}{\ensuremath{\gamma}}

The \master offloads the computation of $\bfC \define \bfA\bfB \in \F_q^{k \times m}$ to $n$ \workers. The workers are \emph{honest-but-curious}, i.e., they will follow the computation protocol but want to obtain insights about the \master's private data. Privacy can be guaranteed by encoding the input matrices using secret sharing with parameters $n$, $t\leq n$, and $z<t$. Each input matrix is split into $t-z$ parts $\bfA_1,\dots,\bfA_{t-z}, \bfB_1,\dots,\bfB_{t-z}$ and encoded with polynomials of the form $\polya[x] = \sum_{\ell=1}^{t-z}\bfA_\ell x^{\powersA_{\ell-1}} + \sum_{j=1}^z\bfR_{j} x^{\powersA_{t-z+j}}$ and $\polyb[x] = \sum_{\ell=1}^{t-z}\bfB_\ell x^{\powersB_{\ell-1}} + \sum_{j=1}^z\bfS_{j} x^{\powersB_{t-z+j}}$, where $\powersA_\ell$ and $\powersB_\ell$ $\forall \ell \in [t]$ are carefully chosen distinct elements from $\mathbb{F}_q^\star$, e.g., \cite{aliasgari2020private,GASP}, and $\bfR_j$ and $\bfS_j$ are matrices with entries drawn independently and uniformly at random from $\mathbb{F}_q$. The compute task of worker $i\in [n]$ is to multiply $\sharea \define \polya[\coeff]$ by $\shareb \define \polyb[\coeff]$, where the $\coeff$'s are distinct non-zero elements of $\mathbb{F}_q$. Any collection of $z$ workers learns nothing about $\bfA$ and $\bfB$; and the choice $\powersA_\ell$'s and $\powersB_\ell$'s guarantees that any collection of $2t-1$ results of the form $\sharea\cdot \shareb$ suffices to obtain $\bfC$. However, the \workers multiply two \emph{dense} matrices irrespective of the sparsity of $\bfA$ and $\bfB$. In this work, we construct secret sharing schemes that assign sparse matrices to the \workers. 

We restrict ourselves to constructing schemes with $z=1$, known as the non-colluding regime in which the \workers do neither communicate nor collaborate to learn the \master's data, and with $t=2$. 
The \master then encodes the input matrices $\bfA$ and $\bfB$ into $n$ computational tasks called \emph{shares} of the form\footnote{This encoding assigns tasks of the same size as $\bfA$ and $\bfB$. To assign smaller tasks to the {\workers}, $\bfA$ and $\bfB$ can be split into smaller sub-matrices and encoded into several degree-$1$ polynomial pairs. Stragglers tolerance is guaranteed by assigning extra evaluations of those polynomial pairs.} $(\sharea, \shareb)$, $i\in [n]$, where $\sharea \define \polya[\coeff] \define \bfA + \coeff\bfR$ and $\shareb[i] \define \polyb[\coeff] \define \bfB + \coeff \bfS$. The random matrices $\bfR$ and $\bfS$ are designed depending on $\bfA$ and $\bfB$, respectively. %
The main goal is to design $\bfR$ and $\bfS$ to satisfy the desired sparsity guarantees on $\polya[\coeff]$ and $\polyb[\coeff]$ for all $\coeff \in [\shares]$ and the required privacy guarantee as defined next.

\emph{Privacy Measure:} %
We say that the share $(\sharea, \shareb)$ leaks $\varepsilon \define \I_q(\rvA;\rvA+\coeff\rvR)$ information about $\bfA$. The same is defined for $\bfB$. The relative leakage is defined as $\relleakage \define \I_q(\rvA;\rvA+\coeff\rvR)/\textrm{H}_q(\rvA)$, i.e., the ratio of information leakage to the entropy of the private matrices. For $\varepsilon = 0$, \emph{perfect} information-theoretic privacy is achieved by secret sharing with the sparsity of all shares being $1/q$. The setting for $\varepsilon>0$ and $n=2$ was considered in \cite{sparse_ISIT,xhemrishi2022efficient}. We examine the setting for $\varepsilon>0$ and $n\geq3$. %
We are interested in the fundamental theoretical tradeoff that exists between the desired sparsity of the shares and the privacy guarantees that can be given. %

\section{Motivating example}

\label{sec:example}
We start with a motivating example highlighting the disadvantage of a straightforward application of the encoding in \cite{xhemrishi2022efficient} to sparse private matrix-matrix multiplication. 
Following the scheme proposed in~\cite{xhemrishi2022efficient}, the \master generates $\bfR \in \F_q^{k \times \ell}$ and $\bfS \in \F_q^{\ell \times m}$, $q\geq 2$, %
and creates four shares $\bfR, \bfA+\bfR, \bfS$ and $\bfB+\bfS$. The computational tasks, sent to four \workers, are%
\begin{align*}
    \bfT_1 &= (\bfA + \bfR) (\bfB + \bfS), \quad & %
    \bfT_2 &=(\bfA + \bfR) \bfS  %
    \\
    \bfT_3 &=\bfR (\bfB + \bfS), \quad & %
    \bfT_4 &=\bfR \bfS\,.  %
\end{align*} \vspace{-0.2cm}
Observe that the \master can reconstruct $\bfC$ by performing
 \begin{equation*}
     \bfC = \bfT_1 - \bfT_2 - \bfT_3 + \bfT_4\,. \vspace{-0.2cm}
 \end{equation*}
To tolerate stragglers, the \master resorts to carefully dividing $\bfA$ and $\bfB$ into block matrices, encoding each block into the aforementioned tasks and using fractional repetition task assignments. 
However, for $q\geq3$, a clever change in the computational tasks decreases the number of required workers from four to only three by using $\bfT_1' = \bfT_1$ shown above and \begin{align*}
    \bfT'_2 &= (\bfA + 2^{-1}{\bfR}) \bfS, &
    \bfT'_3 &= \bfR (\bfB + 2^{-1}{\bfS}), 
    \vspace{-0.4cm}
\end{align*}
where $2^{-1}$ is the multiplicative inverse of $2$ in $\mathbb{F}_q$. %
Note that
\begin{equation*}\vspace{-0.1cm}
    \bfC = \bfT'_1 - \bfT'_2 - \bfT'_3\,.
\end{equation*}
Decreasing the number of workers to three comes at the expense of requiring the shares of the form $\bfA+\bfR$, $\bfR$ and $\bfA+2^{-1}\bfR$ to be sparse and private. Hence, adding more constraints to designing the random matrix $\bfR$. As depicted in~\cref{fig:leak_per_2_3}, those constraints lead to more leakage for the same desired sparsity level than the naive scheme. %

Motivated by this example, we ask the following question: \emph{Can encoding functions of the form $\polya[x] = \bfA+x\bfR$ and $\polyb[x] = \bfB+x\bfS$ be used to tolerate stragglers and maintain sparsity of the shares? If so, how does this encoding affect the privacy guarantee?}

In this work, we answer the question in the affirmative. We show that straggler tolerance via encoding polynomials of the form $\polya[x] = \bfA+x\bfR$ comes at the expense of a degradation of the privacy guarantee as depicted in \cref{fig:leak_per_2_3}. We formulate this problem as a non-linear convex optimization problem, which we solve under mild assumptions explained in the sequel. The optimization results in a sparse secret sharing scheme achieving maximum privacy guarantees for a desired sparsity.

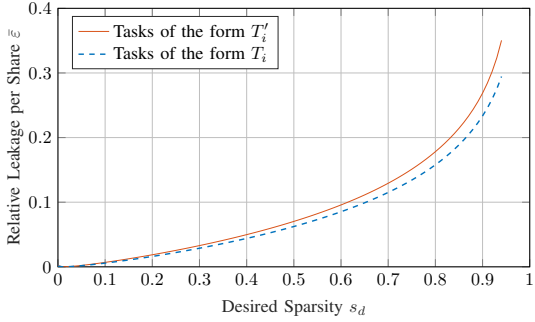
\begin{figure}
    \centering
    \resizebox{.8\linewidth}{!}{\input{tikz/leakage_per_share_simple}}
    \vspace{-1ex}
    \caption{Desired sparsity level of the shares $\sd$ vs. the relative leakage $\relleakage$ per share, for the case where the input matrix has sparsity level $s = 0.95$ and lives in a finite field of cardinality $q=89$.}
    \vspace{-0.5cm}
    \label{fig:leak_per_2_3}
\end{figure}

\section{Straggler-Tolerant Private Distributed Sparse Matrix-Matrix Multiplication} \label{sec:scheme}

The crux of our straggler-tolerant scheme is constructing a family of sparse secret sharing schemes that generalizes the well-known Shamir's secret sharing scheme ($z=t-1$) \cite{S79, Blakley1979} for $t=2$.
The idea is as follows. Given an encoding rule 
\begin{align*}
    \polya &= \bfA + x \bfR, &
    \polyb &= \bfB + x \bfS ,
\end{align*}
we can define a polynomial $\respoly[x]$ as
\begin{equation}\label{eq:h_pol}
    \respoly = \polya \cdot \polyb = \bfA \bfB + x (\bfR \bfB + \bfA \bfS) + x^2 \bfR \bfS\,.
\end{equation}

The polynomial $\respoly[x]$ is a degree-$2$ polynomial that satisfies $\respoly[0] = \bfC = \bfA\bfB$. In order to interpolate $\respoly[x]$, the \master requires three distinct evaluations of $\respoly[x]$ which can be obtained by assigning to the \workers computational tasks of the form $(\polya[\coeff],\polyb[\coeff])$ for $n\geq 3$ distinct symbols $\coeff\in \mathbb{F}_q^\star$ and waiting for three workers to return their results.
If sparsity is of no concern, the entries of $\bfR$ and $\bfS$ can be drawn independently and uniformly at random from $\mathbb{F}_q$ and the encoding would work for any value of $n$.

However, to guarantee sparsity in the shares, the entries of $\bfR$ and $\bfS$ must be drawn carefully and dependently on the values of $\bfA$ and $\bfB$, respectively, and on the number of the desired shares $n$ and the chosen symbols $\coeff \in \mathbb{F}_q^\star$. The choice of $\coeff$ and the distribution from which $\bfR$ and $\bfS$ are generated will affect the privacy guarantee and the sparsity of the shares.

Our main result can be summarized as follows. Motivated by the implication of Lemma~\ref{lemma:sdiff0}, we generate the entries of $\bfR$ according to the distribution shown in \cref{def:model}. The entries of $\bfS$ are generated similarly. We give in \cref{thm:optimal_pmf_straggler_tol} the best privacy guarantees achievable for a desirable sparsity constraint for such random matrices. As a result, we construct optimal sparse secret sharing schemes.

\begin{lemma} \label{lemma:sdiff0}
    Let $\bfA$ be a matrix with independent and identically distributed entries such that the non-zero values are chosen uniformly at random from $\mathbb{F}_q^*$. Consider a scheme as in \cite{xhemrishi2022efficient} (\cref{def:model} for $\shares = 2$) that encodes $\bfA$ into two sparse secret shares $\bfR$ and $\bfA+\bfR$ with sparsity levels $\sr$ and $\sar$, respectively. Then, the minimum leakage for a desired average sparsity $\savg\define\frac{\sar+\sr}{2}$ is obtained for $\sar=\sr =\savg$.
\end{lemma} 

The proof will be given in an extended version of this paper.
\cref{lemma:sdiff0} implies that when designing a sparse Shamir secret sharing scheme for $\shares =2$, the minimum leakage is obtained by requiring both shares to have the same desired sparsity. %
As such, we constrain ourselves to schemes where all $\shares$ shares have the same sparsity level and present \cref{def:model} which generalizes the scheme in \cite{xhemrishi2022efficient}.

\begin{construction}[Conditional Distribution of the Random Matrix]\label{def:model}The entries of the matrix $\bfR$ are drawn according to the following conditional distribution.
\begin{align}
    \Pr\{\rvR = r \lvert \rvA = 0\} \!&\!= \begin{cases} 
    \pz, &r = 0 \\   \frac{1-\pz}{q-1} , &r \neq 0,
    \end{cases} \label{eq:gdependent_on_0_final}\\
    \Pr\{\rvR = r \lvert \rvA = a\} \!&\!= \begin{cases} 
    \pc, &r \in \{-\frac{a}{\coeff}\}_{i \in [\shares]} \\
    \frac{1-\shares \pc}{q-\shares} , &r \not\in \{-\frac{a}{\coeff}\}_{i \in [\shares]}, \\
    \end{cases} \label{eq:gdependent_on_nz_final}
\end{align}
where $r\in \F_q$, $a \in \F_q^*$, and negative numbers correspond to the additive inverses of the absolute values in $\F_q^*$.
\end{construction}
\vspace{-0.3cm}
\cref{def:model} captures the following ideas. The parameter $\pc$ is responsible for tailoring the probability of creating $0$'s in a single share. That is, the probability values used to tweak the sparsity with the privacy of each of the shares $i \in [\shares]$ given a non-zero realization of $\rvA = a$ by choosing $\rvR = -a/\coeff$ are the same for all shares. The second important parameter $\pz$ is responsible for inheriting a $0$ in every share. Particularly, if the realization of $\rvA$ is zero, then the probability of choosing $\rvR=0$ equally affects the sparsity of all shares. Choosing $\pz, \pc > 1/q$ leads to information leakage through the shares, which is why we have to optimize over these parameters to result in a desired sparsity while revealing as little as possible about the private matrix.

To maximize the entropy of the shares and hence minimize the leakage, we choose the remaining values of $\bfR$ (that do not affect the sparsity of any of the shares) to be uniformly distributed among the remaining symbols in the field. Our choice is motivated by the assumption that the non-zero entries of $\bfA$ are uniformly distributed over $\mathbb{F}_q^*$.

We jointly choose $\pz, \pc $ and the evaluation points $\coeff$ of $\polya$ and $\polyb$ such that all shares $\sharea\define \polya[\coeff]$ and $\shareb\define \polyb[\coeff]$ attain a desired level of sparsity defined as $\sda \define \sparsity(\sharea[1]) \! = \! \sparsity(\sharea[2]) \! = \! \cdots \! = \! \sparsity(\sharea[\shares])$ and $\sdb = \sparsity(\shareb[1]) \! = \! \sparsity(\shareb[2]) \! = \! \cdots \! = \! \sparsity(\shareb[\shares])$, respectively. Since the shares of $\bfA$ and $\bfB$ are generated separately, the sparsities of $\bfA$ and $\bfB$ may differ. %

Throughout the rest of the paper, we consider only one instance of sparse secret sharing with desired sparsity $\sd\define \sda$ for the private matrix $\bfA$ and refer to the sparsity level of $\bfA$ as $s\define \slevel{\bfA}$. All results hold equally for $\bfB$ and the generation of $\bfS$, respectively. %

\begin{lemma}\label{lemma:sparsity}
Following \cref{def:model}, every share has the following level of sparsity:
\begin{align*}
    \sd &= \pz s + \pc (1-s).
\end{align*}
\vspace{-6ex}
\end{lemma}
\begin{proof}
    This result is a direct consequence of \eqref{eq:gdependent_on_0_final} and \eqref{eq:gdependent_on_nz_final}.
\end{proof}

The leakage minimizing parameters for $\pz$ and $\pc$ for $\shares$ shares each with sparsity level $\sd$ are given by \cref{thm:optimal_pmf_straggler_tol}, for which we define the following quantities. Let $s_1 \define \sd/(1-s)$, $s_2 \define (s-\sd)/(1-s)$ and $c \define (q-1)/(q-\shares)^\shares$, we have %
\begin{align*}
    b_{\shares+1} &= -1 - c(-\shares)^{\shares} \\
    b_{\shares} &= c \big( s_1 (-\shares)^\shares - \shares (-\shares)^{\shares-1} \big) -s_2 \\
    b_k &= c \left(s_1 \binom{\shares}{k} (-\shares)^k - \binom{\shares}{k-1} (-\shares)^{k-1} \right), \forall k\in [\shares-1] \\
    b_0 &= c s_1.
\end{align*}

\begin{theorem} \label{thm:optimal_pmf_straggler_tol}
Given a desired sparsity $\sd$ of the shares, the optimal PMF of the form given in \eqref{eq:gdependent_on_0_final} and \eqref{eq:gdependent_on_nz_final} that minimizes the leakage of $\shares$ shares is obtained by setting $\pc$ as the real root of the polynomial $\sum_{j = 0}^{n+1} b_j \pc^j$ in $\pc$ that satisfies $0 \leq \pc (1-s) \leq \min\{\sd, \frac{1}{n} \}$ and computing $\pz$ as
\begin{equation*}
    \pz = \frac{\sd - \pc(1-s)}{s}.
\end{equation*}
\end{theorem}

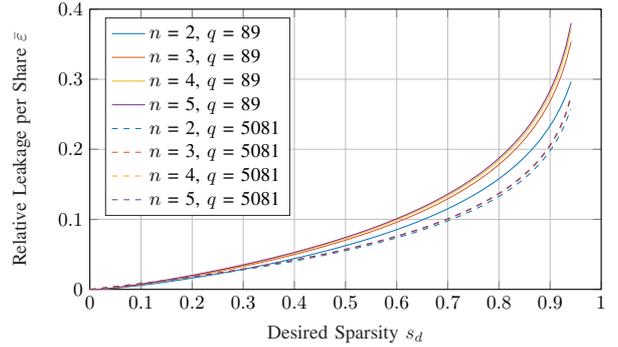
\begin{figure}[t]
    \centering
    \resizebox{.9\linewidth}{!}{\input{tikz/leakage_per_share}}
    \vspace{-1ex}
    \caption{Relative leakage of $\sharea$ and $\shareb$ over their sparsity for different number of shares $\shares$, $s=0.95$ and $q \in \{89, 5081\}$.}
\label{fig:leakage_over_sparsity}
    \vspace{-0.5cm}
\end{figure}

We depict in \cref{fig:leakage_over_sparsity} the relative leakage of each share as a function of $\sd$ and the number $\shares$ of total shares. We consider the finite field $\Fq$ for $q=89$ and $q=5081$. As $\shares$ grows\footnote{This statement holds for $\shares \ll q$. If $\shares$ approaches $q$, the constraints imposed by the optimization problem in \cref{prop:matrix_optimization} lead to non-sparse secret sharing.}, each share leaks more about the private input matrix for a fixed $\sd$. For large values of $q$, the loss of privacy incurred by increasing $\shares$ becomes negligible. %
In addition, for fixed $\shares$ and fixed $\sd$, increasing $q$ decreases the relative leakage per share $\relleakage$. For example, for $q=89$ and $\sd = 0.9$, $\relleakage$ is equal to $0.234$ for $\shares = 2$ and increases to $0.284$ when increasing $\shares$ to $5$. However, for $q=5081$ and $\sd=0.9$, $\relleakage$ increases from $0.199$ to $0.207$ when increasing $\shares$ from $2$ to $5$. Notice that for fixed $\shares$ and $\sd$, increasing $q$ decreases $\relleakage$, cf. \cref{fig:leakage_over_sparsity}.

\section{Analysis of the Trade-off between Sparsity Privacy and Number of Shares} \label{sec:analysis}
In this section, we analyze the sparsity of the shares and the leakage about $\rvA$ for a fixed $\shares$ through $\rvA+\coeff\rvR$ for all $i \in [\shares]$.

The goal is to design the PMF of $\rvR$ given $\rvA$ that minimizes the leakage $\mutinf\left(\rvA+\coeff\rvR; \rvA\right)$ while respecting the desired sparsity $\sd\define\sda$ and the coefficients $\coeff$ for all $i \in [\shares]$ as given in \cref{prop:matrix_optimization}. Designing the conditional PMF boils down to finding $\pra \define \Pr(\rvR = r|\rvA = a)$ for all $r,a \in \mathbb{F}_q$ resulting in a proper PMF and guaranteeing the desired level of sparsity. Hence, we optimize the leakage of the shares over the set $\mathcal{P} \define \{p_{r\vert a}: r,a\in \mathbb{F}_q\}$ while considering the coefficients $\coeff$ of the secret sharing scheme. This is summarized in \cref{prop:matrix_optimization}, which yields an optimization problem for the average leakage through the shares. Expressing mutual information in terms of KL divergence and performing variable substitutions gives the optimization as a function of $\mathcal{P}$.

\begin{proposition} \label{prop:matrix_optimization}
Let $\sd$ be the desired sparsity for each of the secret shares, then the optimal leakage is given by
\begin{align*}
    \lossopt &= \optimizer \frac{1}{\shares} \sum_{i\in [\shares]} \leakage{i} = \optimizer \frac{1}{\shares} \sum_{i \in [\shares]} \mutinf\left(\rvA+\coeff\rvR; \rvA\right) \\
    &=\optimizer \begin{aligned}[t] &\sum_{i \in [\shares]} \kl{\mathrm{P}_{A,\rvA+\coeff\rvR}}{\pa \mathrm{P}_{\rvA+\coeff\rvR}} \end{aligned} \\
    &=\optimizer \!\!\! \sum_{i \in [\shares]} \sum_{a,b\in\mathbf{F}_q} \!\!\! \pa(a) \!\left(\pra[(b-a)/\coeff][a] \log\frac{\pra[(b-a)/\coeff][a]}{\papr(b)} \right)\!
\end{align*}
subject to \vspace{-0.68cm}
\begin{align*}
    \forall a\in\Fq: \pra[0][a] + \sum_{r\in\Fqstar} \pra[r][a] - 1 &= 0, \\[-5pt] %
    \!\forall i \in [\shares]: \pra[0][0] \cdot \pa(0) + \sum_{a\in\Fqstar} \pra[-a/\coeff][a] \cdot \pa(a) -\sd &= 0. \! %
\end{align*}
\end{proposition}
\vspace{-1ex}
\cref{prop:matrix_optimization} describes a non-linear convex optimization problem with affine constraints that can be solved using Lagrange multipliers \cite{rockafellar1993lagrange}. 
To reduce the complexity of the optimization problem, we will abstract the problem to only optimize over values in $\mathcal{P}$ that are of particular interest. Therefore, we rely on the result of \cref{lemma:sdiff0} and \cref{def:model}. %

As a consequence of the conditional PMF in \cref{def:model}, the constraints of the optimization in \cref{prop:matrix_optimization} become
\vspace{-.1cm}
\begin{align}
    c_1(\pz,\pzinv) &\define \pz + (q-1) \pzinv - 1 &&= 0, \label{eq:gconstraint1strag} \\
    c_2(\pc,\pcinv) &\define \shares \pc + (q-\shares) \pcinv - 1 &&= 0, \label{eq:gconstraint2strag} \\
    c_3(\pz,\pext) &\define \pz s + \pc (1-s) - \sd &&= 0, \label{eq:gconstraint3strag}
\end{align}
where $\pz,\pzinv \define (1-\pz)/(q-1)$ and $\pcinv \define (1-\shares \pc)/(q-\shares)$ are non-negative and 1 at most.

\newcommand{\xlogx}[2]{\ensuremath{z(#1, #2)}}
For ease of notation, we define $\xlogx{x}{y} \define x\log(\frac{x}{y})$ for the rest of the paper. We now reformulate the objective function in \cref{prop:matrix_optimization} to conditional PMFs as in \cref{def:model}. Since the PMF is symmetric in the shares, i.e., all shares have the same level of sparsity, instead of the average leakage we express the objective as the leakage of one share.

\begin{lemma} \label{lemma:leakage_strag}
Considering PMFs of the form given in \cref{def:model}, the leakage $\totalleakagestragglers \define \mathrm{L}_i$ for any share $i\in[\shares]$ is given by
\begin{align*}
&\totalleakagestragglers = s \big[ \xlogx{\pz}{\sd} +(q-1)\xlogx{\pzinv}{\sdinv} \big] \\
& \!+ \! (1-s) \big[ \xlogx{\pc}{\sd} \!+\! (\shares-1) \xlogx{\pc}{\sdinv} \!+\! (q-\shares) \xlogx{\pcinv}{\sdinv} \big]\!,
\end{align*}
where $\sdinv$ describes the probability of observing any non-zero entry in the share and is given by $\frac{1-\sd}{q-1}$.
\begin{proof}
\vspace{-2ex}
The statement follows from evaluating and simplifying the objective in \cref{prop:matrix_optimization} according to the construction of the conditional PMF in \eqref{eq:gdependent_on_0_final} and \eqref{eq:gdependent_on_nz_final}.
\end{proof}
\end{lemma}

\vspace{-2ex}
Given the conditional PMF in \cref{def:model} and its impact on the constraints (cf. \eqref{eq:gconstraint1strag}-\eqref{eq:gconstraint3strag}) and the leakage (cf. \cref{lemma:leakage_strag}), %
we show in the following how to solve the non-linear convex optimization problem and hence how to prove \cref{thm:optimal_pmf_straggler_tol}.
\vspace{-2ex}
\begin{proof}[Proof of \cref{thm:optimal_pmf_straggler_tol}]
To find the distribution of $\rvR$ given $\rvA$ that minimizes the information-leakage of $\rvA$ through the shares $\rvA + \coeff \rvR$, we utilize the method of Lagrange multipliers \cite{rockafellar1993lagrange} to combine the objective function, i.e., the leakage from \cref{lemma:leakage_strag} with the constraints in \eqref{eq:gconstraint1strag}-\eqref{eq:gconstraint3strag}. By doing that, the objective function to be minimized can be expressed as
\begin{align}
    \lossabbrev \define & \totalleakagestragglers  + \lambda_1 c_1(\pz,\pzinv) \nonumber \\ %
    &+ \lambda_2 c_2(\pc,\pcinv) + \lambda_3 c_3(\pz,\pc). \label{eq:leakage}
\end{align}

The minimization problem of \cref{prop:matrix_optimization} can now be solved by setting the gradient of the objective in \eqref{eq:leakage} to zero, i.e., $\gradstragglers \lossabbrev = 0$, which amounts to solving a system of seven equations with seven unknowns. 
We first utilize a subset of those equations
and show that given the objective $\lossabbrev$, we obtain from $\nabla_{\pz,\pzinv,\pc,\pcinv} \lossabbrev = 0$ the relation
\begin{align}
    \pz (\pcinv)^\shares &= \pzinv \pc^\shares. \label{eq:grad_obj_compact}
\end{align}
We use that $\frac{\partial \xlogx{x}{y}}{\partial x} = \log(x)-\log(y)+1$. Then, from \eqref{eq:leakage}, \cref{lemma:leakage_strag} and the constraints in \eqref{eq:gconstraint1strag}-\eqref{eq:gconstraint3strag}, we obtain the following non-linear system of equations for $\nabla_{\pz,\pzinv,\pc,\pcinv} \lossabbrev = 0$:
\begin{align*}
    \frac{\partial \lossabbrev}{\partial p_1} &= s\left(\log(p_1) - \log(\sd) + 1 \right) + \lambda_1 + s\lambda_3 = 0, \\
    \frac{\partial \lossabbrev}{\partial p_1^{\text{inv}}} &= s(q-1)\left(\log(p_1^{\text{inv}}) - \log(\sdinv) + 1\right) + (q-1) \lambda_1 = 0, %
    \end{align*}
    \vspace{4ex}
    \begin{align*}
    \frac{\partial \lossabbrev}{\partial \pc} &= \begin{aligned}[t] & (1-s)\Big(\log(\pc) - \log(\sd) + 1 +  (\shares-1)\\
    & \cdot  ( \log (\pc) - \log(\sdinv) + 1 ) \Big) \! + \! \shares \lambda_2 + (1-s) \lambda_3 = 0, \end{aligned} \\
    \frac{\partial \lossabbrev}{\partial \pcinv} &= \begin{aligned}[t] & (1-s)(q-\shares) \left(\log(\pcinv) - \log(\sdinv) + 1\right) \\
    & + (q-\shares) \lambda_2 = 0. \end{aligned}
\end{align*}
By scaling, we simplify and reformulate this system of equations and obtain
\begin{align}
    0 &= \log(p_1) - \log(\sd) + 1 + \lambda_1^\prime + \lambda_3^\prime, \label{eqline:stragpartial_p1} \\
    0 &= -\log(p_1^{\text{inv}}) + \log(\sdinv) - 1 - \lambda_1^\prime, \label{eqline:stragpartial_p1inv} \\
    0 &= \begin{aligned}[t] - \shares \log(\pc) + &\log(\sd) + (\shares-1) \log(\sdinv) \\& - \shares - \shares \lambda_2^\prime - \lambda_3, \end{aligned} \label{eqline:stragpartial_pc} \\
    0 &= \shares \log(\pcinv) - \shares \log(\sdinv) + \shares + \shares \lambda_2^\prime, \label{eqline:stragpartial_pcinv}
\end{align}
where $\lambda_1^\prime = \lambda_1/s$, $\lambda_2^\prime = \lambda_2/(1-s)$ and $\lambda_3^\prime = \lambda_3/s$. 
Summing \eqref{eqline:stragpartial_p1}-\eqref{eqline:stragpartial_pcinv} 
leads to $\log\left(\pz {\pcinv}^\shares\right) - \log\left(\pzinv \pc^\shares\right)=0$. Exponentiating both sides yields the result stated in \eqref{eq:grad_obj_compact}. The condition on the root $\pc$ is a consequence of \eqref{eq:gconstraint2strag} and \eqref{eq:gconstraint3strag}.

We now utilize \eqref{eq:grad_obj_compact} and $\nabla_{\lambda_1, \lambda_2, \lambda_3} \lossabbrev = 0$, where the latter system of equations is linear and equivalent to \eqref{eq:gconstraint1strag}-\eqref{eq:gconstraint3strag}. We solve \eqref{eq:gconstraint1strag}, \eqref{eq:gconstraint2strag} and \eqref{eq:gconstraint3strag} for $\pz,\pzinv,\pcinv$ as a function of $\pc$ to obtain
\vspace{-.1cm}
\begin{align}
    (q-1) \frac{\sd-(1-s)\pc}{s-\sd+(1-s)\pc} = \left((q-\shares)\frac{\pc}{1-\shares\pc}\right)^\shares\!\!\!\!. \label{eq:grad_obj_compact_stragglers}
\end{align}
We reformulate \eqref{eq:grad_obj_compact_stragglers} to get the following polynomial in $\pc$ with degree $\shares+1$:
\vspace{-.2cm}
\begin{align*}
    -\pc^{\shares+1} - \! s_2 \pc^\shares + c s_1 \! \sum_{k=1}^\shares \binom{\shares}{k} (-\shares \pc)^k - \pc c \! \sum_{k=0}^\shares (-\shares \pc)^k = 0.
\end{align*}

\vspace{-.1cm}
\noindent Sorting the coefficients by powers of $\pc$ results in the root finding problem given in \cref{thm:optimal_pmf_straggler_tol}. Due to convexity, the solution is unique, i.e., local and global minimizers are the same.
The results for $\pz,\pzinv$ and $\pcinv$ can be obtained from $\nabla_{\lambda_1,\lambda_2,\lambda_3} \lossabbrev = 0$, i.e., from the constraints stated in \eqref{eq:gconstraint1strag}-\eqref{eq:gconstraint3strag}.
\end{proof}

\section{Conclusion}
We considered the problem of outsourcing the multiplication of two private and sparse matrices to untrusted workers. To guarantee straggler tolerance, privacy, and sparsity in the computational tasks, we designed a generalization of Shamir's secret sharing scheme for a threshold $t=2$ that outputs shares with a desired sparsity. We showed that the price of sparsity and straggler tolerance is a weaker privacy guarantee. More precisely, for a desired sparsity level, increasing the straggler tolerance by increasing the number of shares $\shares$ results in fundamentally weaker privacy guarantees. %
Nevertheless, the privacy degradation becomes negligible and even vanishes when fixing $\shares$ and increasing the field size $q\gg n$. Future work includes the generalization to a sparse McEliece-Sarwate \cite{McESa81} secret sharing scheme, i.e., secret sharing with threshold $t\geq 3$ and an arbitrary number of colluding workers $1\leq z\leq t-1$. This case is challenging since the padding matrices must be drawn dependently on each other and all secret inputs simultaneously.

\clearpage
\pagebreak
\bibliographystyle{IEEEtran}
\bibliography{IEEEabrv,literature}

\nobalance

\input{appendix}
\clearpage

\end{document}

%% file: packages.tex
\usepackage{pgfgantt}
\usepackage[utf8]{inputenc} 
\usepackage[T1]{fontenc}
\usepackage{url}
\usepackage[english]{babel}
\usepackage{graphicx}
\usepackage{amssymb}
\usepackage[cmex10]{amsmath}
\usepackage{amsthm}
\usepackage{amsfonts}
\usepackage[capitalize]{cleveref}
\usepackage[inline]{enumitem}
\usepackage{cite}
\usepackage{epsfig}
\usepackage{graphics}
\usepackage{array}
\usepackage{color}
\usepackage{xcolor}
\usepackage{multirow}
\usepackage[normalem]{ulem}
\usepackage{arydshln}
\usepackage[makeroom]{cancel}
\usepackage{varwidth}
\usepackage{pgfplots}
\usepackage{nicefrac}
\usepackage{adjustbox}
\usepackage[ruled,vlined]{algorithm}
\usepackage{algpseudocode}
\usepackage{colortbl}
\usepackage{booktabs}
\usepackage{float}
\usepackage{multicol}
\usepackage{ctable}
\usepackage{xspace}
\usepackage{subfig}
\usepackage{siunitx}
\usepackage{balance}
\usepackage{xargs}
\usepackage{ifthen}
\usepackage[
font=small,labelfont=bf
]{caption}
\usepackage{flushend}
\flushend
\usepackage{bbm}
\usepackage{diagbox}
\usepackage{dsfont}
\usepackage[nolist]{acronym}
\usepackage{mathtools} %
\usepackage{tcolorbox, lipsum}
\tcbuselibrary{skins}
\tcbuselibrary{breakable}
\usepackage{tikz}
\usetikzlibrary{backgrounds,calc,shadings,shapes.arrows,shapes.symbols,shadows,fit,positioning,topaths,hobby,shapes.geometric,shapes,snakes,patterns,arrows,automata,plotmarks}
\algnewcommand{\Initialize}{%
  \State \textbf{Initialize:}
}
\usepackage{url}

\usepackage[
    shortcuts,       %
    acronym,         %
    section=section, %
]{glossaries}
\makeglossaries
\usepackage{setspace}

\renewenvironment{thebibliography}[1]{
  \begin{oldthebibliography}{#1}
    \setlength{\itemsep}{.12em}
    \setlength{\parskip}{.1em}
}
{
  \end{oldthebibliography}
}

%% file: commands.tex
\definecolor{green1}{rgb}{0,0.5,0}
\definecolor{magenta}{rgb}{1.0, 0.11, 0.81}
\definecolor{mulberry}{rgb}{0.77, 0.29, 0.55}
\definecolor{xgray}{rgb}{0.9, 0.9, 0.9}

\def \bes{\begin{equation*}}
\def \ees{\end{equation*}}
\def \bas{\begin{align*}}
\def \eas{\end{align*}}
\def \be{\begin{equation}}
\def \ee{\end{equation}}
\def \bbm{\begin{bmatrix}}
\def \ebm{\end{bmatrix}}

\def \rvA{\texttt{A}}
\def \rvB{\texttt{B}}

\def \rvR {\texttt{R}}

\def \F{\mathbb{F}}

\newcommand{\Aij}{\rvA_{\{i,j\}}}
\newcommand{\Rij}{\rvR_{\{i,j\}}}

\newcommand{\bfA}{{\mathbf A}}
\newcommand{\bfB}{{\mathbf B}}
\newcommand{\bfC}{{\mathbf C}}

\newcommand{\bfR}{{\mathbf R}}
\newcommand{\bfS}{{\mathbf S}}
\newcommand{\bfT}{{\mathbf T}}

\newcommand{\I}{\mathrm{I}}
\newcommand{\mutinf}{\I_q}
\newcommand{\kl}[2]{\mathrm{D}_{\textrm{KL}} \left(#1 \Vert #2 \right)}

\newcommand{\leakage}[1]{\mathrm{L}_#1}

\newcommand{\slevel}[1]{\ensuremath{s_{#1}}}
\newcommand{\sr}{\ensuremath s_{\bfR}}
\newcommand{\sar}{\ensuremath s_{\bfA+\bfR}}

\newcommand{\savg}{\ensuremath s_{\text{avg}}}
\newcommand{\sdiff}{\ensuremath s_{\delta}}

\newcommand{\sparsity}{\textrm{S}}
\newcommand{\pz}{p_{1}}
\newcommand{\pzinv}{p_{1}^{\text{inv}}}
\newcommand{\pnz}{p_{3}}
\newcommand{\pext}{p_{2}}
\newcommand{\pnzinv}{p_{2,3,4}^\text{inv}}
\newcommand{\ptwoinv}{p_{2,3}^\text{inv}}

\newcommand{\master}{main node\xspace}

\newcommand{\workers}{workers\xspace}

\newtheorem{theorem}{Theorem}

\newtheorem{lemma}{Lemma}

\newtheorem{construction}{Construction}
\newtheorem{proposition}{Proposition}
\newtheorem{definition}{Definition}

\usetikzlibrary{decorations.pathreplacing,calc}
\tikzset{brace/.style={decorate, decoration={brace}},
 brace mirrored/.style={decorate, decoration={brace,mirror}},
}

\newcounter{brace}
\newcommand{\Fq}{\mathbb{F}_q}
\newcommand{\Fqstar}{\Fq^\star} 
\newcommand{\pa}[1][a]{\mathrm{P}_{\rvA}}
\newcommand{\paandapr}{\mathrm{P}_{\rvA,\rvA+\rvR}}

\newcommand{\papr}{\mathrm{P}_{\rvA+\coeff\rvR}}

\newcommand{\paandr}[1][a]{\mathrm{P}_{\rvA,\rvR}}
\newcommand{\pr}[1][r]{\mathrm{P}_{\rvR}}
\newcommand{\optimizer}{\min\limits_{\mathcal{P}}} %
\newcommand{\optimizerdelta}{\min\limits_{\mathcal{P},\sdiff}} %
\newcommandx{\pra}[2][1=r, 2=a]{p_{#1\vert#2}}
\newcommand{\lossopt}{\mathrm{L}_\text{opt}}
\newcommand{\totalleakagesdiff}{\mathrm{L}(\pz,\pzinv,\pext,\pnz,\pnzinv,\sdiff)}

\newcommand{\define}{\ensuremath{:=}}

\newcommand{\lossabbrev}{\mathcal{L}}
\newcommand{\grad}{\nabla_{\pz,\pzinv,\pnz,\pext,p_4,\pnzinv,\lambda_1,\lambda_2,\lambda_3,\lambda_4,\lambda_5}}
\newcommand{\gradstragglers}{\nabla_{\pz,\pzinv,\pc,\pcinv,\lambda_1,\lambda_2,\lambda_3}}
\newcommand{\srinv}{\sr^\text{inv}}
\newcommand{\sarinv}{\sar^\text{inv}}

\newcommand{\sd}{\ensuremath{s_d}}
\newcommand{\sda}{\ensuremath{s_{d,\bfA}}}
\newcommand{\sdb}{\ensuremath{s_{d,\bfB}}}
\newcommand{\sdinv}{\ensuremath{s_d^\text{inv}}}

\newcommand{\pc}{\ensuremath{p_\star}}
\newcommand{\pcinv}{\ensuremath{p_\star^\text{inv}}}
\newcommand{\totalleakagestragglers}{\mathrm{L}(\pz,\pzinv,\pc,\pcinv)}
\newcommand{\sharea}[1][i]{\ensuremath{\mathbf{F}_{#1}}}

\newcommand{\shareb}[1][i]{\ensuremath{\mathbf{G}_{#1}}}

\newcommand{\polya}[1][x]{\ensuremath{\mathbf{f}(#1)}}
\newcommand{\polyb}[1][x]{\ensuremath{\mathbf{g}(#1)}}
\newcommand{\respoly}[1][x]{\ensuremath{\mathbf{h}(#1)}}
\newcommand{\eval}[1][i]{\ensuremath{\alpha_{#1}}}

\newcommand{\shares}{\ensuremath{n}}
\newcommand{\coeff}{\ensuremath{\eval}}

\newcommand{\relleakage}{\ensuremath{\bar{\varepsilon}}}

%% file: tikz/leakage_per_share_simple.tex
\definecolor{mycolor1}{rgb}{0.00000,0.44700,0.74100}%
\definecolor{mycolor2}{rgb}{0.85000,0.32500,0.09800}%
\begin{tikzpicture}

\begin{axis}[%
width=3.633in,
height=2.007in,
at={(0.609in,0.402in)},
scale only axis,
xmin=0,
xmax=1,
xlabel style={font=\color{white!15!black}},
xlabel={Desired Sparsity $s_{d}$},
ymin=0,
ymax=0.4,
ylabel style={font=\color{white!15!black}, yshift = -.4cm},
ylabel={Relative Leakage per Share $\relleakage$},
axis background/.style={fill=white},
xmajorgrids,
ymajorgrids,
legend style={at={(0.03,0.97)}, anchor=north west, legend cell align=left, align=left, draw=white!15!black}
]

\addplot [color=mycolor2]
  table[row sep=crcr]{%
0.000196811651249754	0.00180019440366317\\
0.0101968116512498	4.23366220106419e-06\\
0.0201968116512498	0.000211499620962855\\
0.0301968116512498	0.000731836171384254\\
0.0401968116512498	0.00140802488957583\\
0.0501968116512498	0.00218346483885121\\
0.0601968116512498	0.00303066536020181\\
0.0701968116512498	0.00393404188929805\\
0.0801968116512498	0.00488387545388441\\
0.0901968116512498	0.00587370397985351\\
0.10019681165125	0.00689903319333268\\
0.11019681165125	0.00795663557736405\\
0.12019681165125	0.00904414068607386\\
0.13019681165125	0.0101597819027646\\
0.14019681165125	0.0113022326958671\\
0.15019681165125	0.0124704967739389\\
0.16019681165125	0.0136638321136644\\
0.17019681165125	0.01488169705569\\
0.18019681165125	0.0161237112286752\\
0.19019681165125	0.0173896267089858\\
0.20019681165125	0.0186793064160728\\
0.21019681165125	0.0199927077331552\\
0.22019681165125	0.0213298699754698\\
0.23019681165125	0.0226909047429996\\
0.24019681165125	0.0240759884726909\\
0.25019681165125	0.0254853566954285\\
0.26019681165125	0.0269192996357022\\
0.27019681165125	0.0283781588860047\\
0.28019681165125	0.0298623249558161\\
0.29019681165125	0.0313722355446818\\
0.30019681165125	0.0329083744257739\\
0.31019681165125	0.0344712708541489\\
0.32019681165125	0.0360614994352241\\
0.33019681165125	0.0376796804055802\\
0.34019681165125	0.0393264802913541\\
0.35019681165125	0.0410026129201191\\
0.36019681165125	0.0427088407709609\\
0.37019681165125	0.0444459766549237\\
0.38019681165125	0.0462148857245597\\
0.39019681165125	0.048016487817198\\
0.40019681165125	0.0498517601420835\\
0.41019681165125	0.0517217403268395\\
0.42019681165125	0.0536275298440234\\
0.43019681165125	0.0555702978439922\\
0.44019681165125	0.0575512854260378\\
0.45019681165125	0.059571810385941\\
0.46019681165125	0.0616332724848884\\
0.47019681165125	0.0637371592922712\\
0.48019681165125	0.065885052663424\\
0.49019681165125	0.0680786359230877\\
0.50019681165125	0.0703197018365755\\
0.51019681165125	0.0726101614635399\\
0.52019681165125	0.0749520540043091\\
0.53019681165125	0.0773475577663698\\
0.54019681165125	0.0797990023993159\\
0.55019681165125	0.0823088825710604\\
0.56019681165125	0.084879873287198\\
0.57019681165125	0.0875148470900787\\
0.58019681165125	0.0902168934156999\\
0.59019681165125	0.0929893404365058\\
0.60019681165125	0.0958357797786651\\
0.61019681165125	0.0987600945758072\\
0.62019681165125	0.101766491410955\\
0.63019681165125	0.104859536808481\\
0.64019681165125	0.108044199073918\\
0.65019681165125	0.111325896448288\\
0.66019681165125	0.114710552754588\\
0.67019681165125	0.118204661979421\\
0.68019681165125	0.121815363568716\\
0.69019681165125	0.125550530645159\\
0.70019681165125	0.129418873905953\\
0.71019681165125	0.133430064673919\\
0.72019681165125	0.1375948815093\\
0.73019681165125	0.141925386023423\\
0.74019681165125	0.14643513518137\\
0.75019681165125	0.151139439601068\\
0.76019681165125	0.156055680387154\\
0.77019681165125	0.161203701229487\\
0.78019681165125	0.166606298375183\\
0.79019681165125	0.172289839457244\\
0.80019681165125	0.178285054295157\\
0.81019681165125	0.184628058692366\\
0.82019681165125	0.191361699248539\\
0.83019681165125	0.198537348844359\\
0.84019681165125	0.206217348391936\\
0.85019681165125	0.214478397992806\\
0.86019681165125	0.223416382130543\\
0.87019681165125	0.233153432055816\\
0.88019681165125	0.243848613977697\\
0.89019681165125	0.25571476855144\\
0.90019681165125	0.269046387376769\\
0.91019681165125	0.284268737667368\\
0.92019681165125	0.302031822188085\\
0.93019681165125	0.323411580444017\\
0.94019681165125	0.350419499729562\\
};
\addlegendentry{Tasks of the form $T_i^\prime$}

\addplot [color=mycolor1, dashed, thick]
  table[row sep=crcr]{%
0.000196811651249754	0.00104356458553331\\
0.0101968116512498	3.12943990193627e-06\\
0.0201968116512498	0.000162780831708812\\
0.0301968116512498	0.000576207970147309\\
0.0401968116512498	0.00112600809287077\\
0.0501968116512498	0.00176671622869742\\
0.0601968116512498	0.00247516199629072\\
0.0701968116512498	0.00323772046880563\\
0.0801968116512498	0.00404563635338079\\
0.0901968116512498	0.00489293593989038\\
0.10019681165125	0.00577536667028682\\
0.11019681165125	0.00668980691463892\\
0.12019681165125	0.00763391401693252\\
0.13019681165125	0.00860590281468956\\
0.14019681165125	0.00960440011055302\\
0.15019681165125	0.0106283456065165\\
0.16019681165125	0.0116769224554749\\
0.17019681165125	0.0127495073603729\\
0.18019681165125	0.0138456339659344\\
0.19019681165125	0.0149649655276773\\
0.20019681165125	0.0161072742061976\\
0.21019681165125	0.0172724251910615\\
0.22019681165125	0.0184603644116066\\
0.23019681165125	0.0196711089579194\\
0.24019681165125	0.0209047395828201\\
0.25019681165125	0.0221613948265384\\
0.26019681165125	0.02344126642581\\
0.27019681165125	0.0247445957549116\\
0.28019681165125	0.0260716711084058\\
0.29019681165125	0.0274228256812255\\
0.30019681165125	0.028798436135985\\
0.31019681165125	0.0301989216733688\\
0.32019681165125	0.0316247435414128\\
0.33019681165125	0.03307640493505\\
0.34019681165125	0.0345544512496702\\
0.35019681165125	0.0360594706623973\\
0.36019681165125	0.0375920950230215\\
0.37019681165125	0.0391530010434503\\
0.38019681165125	0.040742911780544\\
0.39019681165125	0.0423625984125329\\
0.40019681165125	0.0440128823141388\\
0.41019681165125	0.0456946374401458\\
0.42019681165125	0.0474087930317455\\
0.43019681165125	0.049156336664537\\
0.44019681165125	0.0509383176618566\\
0.45019681165125	0.0527558509021153\\
0.46019681165125	0.0546101210543127\\
0.47019681165125	0.0565023872819149\\
0.48019681165125	0.0584339884619842\\
0.49019681165125	0.0604063489740726\\
0.50019681165125	0.0624209851220245\\
0.51019681165125	0.0644795122618122\\
0.52019681165125	0.066583652720032\\
0.53019681165125	0.0687352446011135\\
0.54019681165125	0.0709362515969789\\
0.55019681165125	0.0731887739313645\\
0.56019681165125	0.0754950605928108\\
0.57019681165125	0.0778575230362364\\
0.58019681165125	0.0802787505638662\\
0.59019681165125	0.0827615276332674\\
0.60019681165125	0.0853088533847376\\
0.61019681165125	0.0879239637340549\\
0.62019681165125	0.09061035644196\\
0.63019681165125	0.0933718196514962\\
0.64019681165125	0.0962124644822628\\
0.65019681165125	0.099136762391509\\
0.66019681165125	0.102149588162043\\
0.67019681165125	0.10525626956439\\
0.68019681165125	0.108462644976322\\
0.69019681165125	0.111775130541266\\
0.70019681165125	0.115200798827624\\
0.71019681165125	0.118747471440128\\
0.72019681165125	0.122423828668149\\
0.73019681165125	0.126239540084614\\
0.74019681165125	0.130205421102955\\
0.75019681165125	0.134333621958163\\
0.76019681165125	0.138637857544362\\
0.77019681165125	0.143133689224104\\
0.78019681165125	0.147838873431826\\
0.79019681165125	0.152773797089302\\
0.80019681165125	0.157962027243634\\
0.81019681165125	0.163431013035615\\
0.82019681165125	0.169212993873581\\
0.83019681165125	0.175346191408375\\
0.84019681165125	0.18187639941308\\
0.85019681165125	0.188859143315447\\
0.86019681165125	0.196362674816997\\
0.87019681165125	0.204472224397093\\
0.88019681165125	0.213296209052346\\
0.89019681165125	0.222975593310566\\
0.90019681165125	0.233698563732722\\
0.91019681165125	0.245724646497749\\
0.92019681165125	0.259426754253781\\
0.93019681165125	0.275370284433525\\
0.94019681165125	0.29447797278307\\
};
\addlegendentry{Tasks of the form $T_i$}

\end{axis}
\end{tikzpicture}%

%% file: tikz/leakage_per_share.tex
\definecolor{mycolor1}{rgb}{0.00000,0.44700,0.74100}%
\definecolor{mycolor2}{rgb}{0.85000,0.32500,0.09800}%
\definecolor{mycolor3}{rgb}{0.92900,0.69400,0.12500}%
\definecolor{mycolor4}{rgb}{0.49400,0.18400,0.55600}%
\begin{tikzpicture}

\begin{axis}[%
width=3.633in,
height=2.007in,
at={(0.609in,0.402in)},
scale only axis,
xmin=0,
xmax=1,
xlabel style={font=\color{white!15!black}},
xlabel={Desired Sparsity $s_{d}$},
ymin=-1.81630948225826e-17,
ymax=0.4,
ylabel style={font=\color{white!15!black}},
ylabel={Relative Leakage per Share $\relleakage$},
axis background/.style={fill=white},
xmajorgrids,
ymajorgrids,
legend style={at={(0.03,0.97)}, anchor=north west, legend cell align=left, align=left, draw=white!15!black}
]
\addplot [color=mycolor1]
  table[row sep=crcr]{%
0.0112359550561798	4.87302544020534e-18\\
0.0212359550561798	0.000197056644436058\\
0.0312359550561798	0.000628063783746169\\
0.0412359550561798	0.00118886569177159\\
0.0512359550561798	0.00183746834421629\\
0.0612359550561798	0.00255206324826557\\
0.0712359550561798	0.00331968296347492\\
0.0812359550561798	0.0041319335776835\\
0.0912359550561798	0.00498305868352333\\
0.10123595505618	0.00586894536466944\\
0.11123595505618	0.00678656646573339\\
0.12123595505618	0.00773364594535367\\
0.13123595505618	0.00870844724711498\\
0.14123595505618	0.00970963368640835\\
0.15123595505618	0.0107361730965424\\
0.16123595505618	0.0117872707973777\\
0.17123595505618	0.0128623213189763\\
0.18123595505618	0.0139608729154904\\
0.19123595505618	0.0150826010282148\\
0.20123595505618	0.0162272881539042\\
0.21123595505618	0.0173948083917574\\
0.22123595505618	0.0185851154716422\\
0.23123595505618	0.019798233417167\\
0.24123595505618	0.0210342492352236\\
0.25123595505618	0.0222933071881814\\
0.26123595505618	0.0235756043207414\\
0.27123595505618	0.0248813869963955\\
0.28123595505618	0.0262109482586874\\
0.29123595505618	0.0275646258769433\\
0.30123595505618	0.028942800969393\\
0.31123595505618	0.0303458971218573\\
0.32123595505618	0.0317743799396099\\
0.33123595505618	0.033228756985215\\
0.34123595505618	0.0347095780672083\\
0.35123595505618	0.0362174358542669\\
0.36123595505618	0.0377529667975716\\
0.37123595505618	0.0393168523508859\\
0.38123595505618	0.0409098204838041\\
0.39123595505618	0.0425326474888894\\
0.40123595505618	0.0441861600883129\\
0.41123595505618	0.0458712378502234\\
0.42123595505618	0.0475888159296343\\
0.43123595505618	0.0493398881532036\\
0.44123595505618	0.0511255104720795\\
0.45123595505618	0.0529468048120411\\
0.46123595505618	0.0548049633557071\\
0.47123595505618	0.0567012532976454\\
0.48123595505618	0.0586370221200347\\
0.49123595505618	0.0606137034442269\\
0.50123595505618	0.0626328235223296\\
0.51123595505618	0.0646960084430463\\
0.52123595505618	0.0668049921377124\\
0.53123595505618	0.0689616252860824\\
0.54123595505618	0.0711678852374007\\
0.55123595505618	0.0734258870810526\\
0.56123595505618	0.0757378960233031\\
0.57123595505618	0.0781063412529844\\
0.58123595505618	0.0805338315104497\\
0.59123595505618	0.0830231726117861\\
0.60123595505618	0.085577387225652\\
0.61123595505618	0.0881997372549641\\
0.62123595505618	0.0908937492423567\\
0.63123595505618	0.0936632432998014\\
0.64123595505618	0.0965123661628443\\
0.65123595505618	0.0994456290934886\\
0.66123595505618	0.102467951509286\\
0.67123595505618	0.105584711408093\\
0.68123595505618	0.108801803899424\\
0.69123595505618	0.112125709459222\\
0.70123595505618	0.115563573915339\\
0.71123595505618	0.119123302673236\\
0.72123595505618	0.122813672342849\\
0.73123595505618	0.126644463780149\\
0.74123595505618	0.130626621683404\\
0.75123595505618	0.13477244738792\\
0.76123595505618	0.139095833532812\\
0.77123595505618	0.14361255204607\\
0.78123595505618	0.14834061073122\\
0.79123595505618	0.153300699124218\\
0.80123595505618	0.15851675196515\\
0.81123595505618	0.164016669756678\\
0.82123595505618	0.169833252315792\\
0.83123595505618	0.176005426005517\\
0.84123595505618	0.182579883568352\\
0.85123595505618	0.189613316025896\\
0.86123595505618	0.197175514823413\\
0.87123595505618	0.205353788821662\\
0.88123595505618	0.21425943234161\\
0.89123595505618	0.224037514932445\\
0.90123595505618	0.234882296797399\\
0.91123595505618	0.247062704214107\\
0.92123595505618	0.260967055322818\\
0.93123595505618	0.277187975895357\\
0.94123595505618	0.296701651650372\\
};
\addlegendentry{$n$ = 2, $q$ = 89}

\addplot [color=mycolor2]
  table[row sep=crcr]{%
0.0112359550561798	1.55050809461079e-17\\
0.0212359550561798	0.000255296746788657\\
0.0312359550561798	0.000796207701318408\\
0.0412359550561798	0.00148459943671602\\
0.0512359550561798	0.00226849401328572\\
0.0612359550561798	0.00312212562693677\\
0.0712359550561798	0.00403071162417457\\
0.0812359550561798	0.00498496065909146\\
0.0912359550561798	0.0059786613798124\\
0.10123595505618	0.00700747718299838\\
0.11123595505618	0.00806828523613734\\
0.12123595505618	0.00915878769873016\\
0.13123595505618	0.0102772701668718\\
0.14123595505618	0.0114224448333927\\
0.15123595505618	0.0125933449120425\\
0.16123595505618	0.0137892514084103\\
0.17123595505618	0.015009641036678\\
0.18123595505618	0.0162541483872895\\
0.19123595505618	0.0175225379574425\\
0.20123595505618	0.0188146831699132\\
0.21123595505618	0.02013055044918\\
0.22123595505618	0.0214701870285829\\
0.23123595505618	0.0228337115596573\\
0.24123595505618	0.0242213068618721\\
0.25123595505618	0.0256332143341131\\
0.26123595505618	0.0270697296771609\\
0.27123595505618	0.0285311996673169\\
0.28123595505618	0.0300180197869141\\
0.29123595505618	0.0315306325655833\\
0.30123595505618	0.0330695265219229\\
0.31123595505618	0.0346352356222623\\
0.32123595505618	0.0362283391939661\\
0.33123595505618	0.0378494622468917\\
0.34123595505618	0.0394992761694711\\
0.35123595505618	0.0411784997763117\\
0.36123595505618	0.0428879006928445\\
0.37123595505618	0.0446282970699265\\
0.38123595505618	0.0464005596277536\\
0.39123595505618	0.0482056140342992\\
0.40123595505618	0.0500444436289823\\
0.41123595505618	0.051918092507567\\
0.42123595505618	0.0538276689896252\\
0.43123595505618	0.0557743494953547\\
0.44123595505618	0.0577593828643326\\
0.45123595505618	0.0597840951550324\\
0.46123595505618	0.0618498949707913\\
0.47123595505618	0.0639582793655928\\
0.48123595505618	0.0661108403916596\\
0.49123595505618	0.0683092723607541\\
0.50123595505618	0.070555379902397\\
0.51123595505618	0.0728510869153792\\
0.52123595505618	0.0751984465242285\\
0.53123595505618	0.0775996521702148\\
0.54123595505618	0.080057049987561\\
0.55123595505618	0.082573152640458\\
0.56123595505618	0.085150654826093\\
0.57123595505618	0.087792450684221\\
0.58123595505618	0.0905016533961571\\
0.59123595505618	0.0932816173070463\\
0.60123595505618	0.0961359629669407\\
0.61123595505618	0.099068605561187\\
0.62123595505618	0.102083787292257\\
0.63123595505618	0.10518611438768\\
0.64123595505618	0.10838059954778\\
0.65123595505618	0.111672710819661\\
0.66123595505618	0.115068428099887\\
0.67123595505618	0.118574308740123\\
0.68123595505618	0.122197564074505\\
0.69123595505618	0.125946149127294\\
0.70123595505618	0.129828868325249\\
0.71123595505618	0.133855500773446\\
0.72123595505618	0.138036949614535\\
0.73123595505618	0.142385421262148\\
0.74123595505618	0.146914641996238\\
0.75123595505618	0.151640121700083\\
0.76123595505618	0.156579477651494\\
0.77123595505618	0.161752835619363\\
0.78123595505618	0.16718333161128\\
0.79123595505618	0.17289774631389\\
0.80123595505618	0.178927316889686\\
0.81123595505618	0.185308789464594\\
0.82123595505618	0.192085803845605\\
0.83123595505618	0.199310745629089\\
0.84123595505618	0.207047270134359\\
0.85123595505618	0.215373815966706\\
0.86123595505618	0.224388618052612\\
0.87123595505618	0.234217068545677\\
0.88123595505618	0.245022899544371\\
0.89123595505618	0.257025883995084\\
0.90123595505618	0.270531308839688\\
0.91123595505618	0.285982302636947\\
0.92123595505618	0.304060929628617\\
0.93123595505618	0.325907783021739\\
0.94123595505618	0.353690784639797\\
};
\addlegendentry{$n$ = 3, $q$ = 89}

\addplot [color=mycolor3]
  table[row sep=crcr]{%
0.0112359550561798	-1.81630948225826e-17\\
0.0212359550561798	0.000283083329617593\\
0.0312359550561798	0.00087354282141229\\
0.0412359550561798	0.00161704231655508\\
0.0512359550561798	0.00245751124907678\\
0.0612359550561798	0.00336779699534966\\
0.0712359550561798	0.00433259669677555\\
0.0812359550561798	0.00534243205659485\\
0.0912359550561798	0.00639104682294747\\
0.10123595505618	0.00747412327029574\\
0.11123595505618	0.00858858623927912\\
0.12123595505618	0.00973219780091824\\
0.13123595505618	0.0109033076052031\\
0.14123595505618	0.0121006920529128\\
0.15123595505618	0.0133234468067608\\
0.16123595505618	0.0145709127254599\\
0.17123595505618	0.0158426235086324\\
0.18123595505618	0.0171382678881073\\
0.19123595505618	0.0184576618323495\\
0.20123595505618	0.019800727810441\\
0.21123595505618	0.0211674791414653\\
0.22123595505618	0.0225580080798383\\
0.23123595505618	0.0239724766957584\\
0.24123595505618	0.0254111098833668\\
0.25123595505618	0.0268741900158821\\
0.26123595505618	0.0283620528968875\\
0.27123595505618	0.029875084748918\\
0.28123595505618	0.031413720046652\\
0.29123595505618	0.0329784400503794\\
0.30123595505618	0.0345697719313099\\
0.31123595505618	0.0361882884073168\\
0.32123595505618	0.0378346078284263\\
0.33123595505618	0.0395093946674883\\
0.34123595505618	0.0412133603842702\\
0.35123595505618	0.0429472646415931\\
0.36123595505618	0.0447119168607685\\
0.37123595505618	0.046508178110955\\
0.38123595505618	0.0483369633335543\\
0.39123595505618	0.0501992439086746\\
0.40123595505618	0.0520960505762699\\
0.41123595505618	0.0540284767299902\\
0.42123595505618	0.0559976821072524\\
0.43123595505618	0.058004896904723\\
0.44123595505618	0.0600514263544162\\
0.45123595505618	0.0621386558021695\\
0.46123595505618	0.0642680563374718\\
0.47123595505618	0.0664411910317168\\
0.48123595505618	0.0686597218511176\\
0.49123595505618	0.0709254173209994\\
0.50123595505618	0.0732401610302492\\
0.51123595505618	0.0756059610787224\\
0.52123595505618	0.0780249605867281\\
0.53123595505618	0.0804994494048744\\
0.54123595505618	0.0830318771851397\\
0.55123595505618	0.0856248680007478\\
0.56123595505618	0.0882812367342244\\
0.57123595505618	0.0910040074909701\\
0.58123595505618	0.093796434341262\\
0.59123595505618	0.0966620247485249\\
0.60123595505618	0.0996045661082651\\
0.61123595505618	0.102628155903076\\
0.62123595505618	0.105737236078257\\
0.63123595505618	0.108936632364567\\
0.64123595505618	0.112231599425525\\
0.65123595505618	0.115627872894532\\
0.66123595505618	0.119131729602323\\
0.67123595505618	0.122750057591954\\
0.68123595505618	0.126490437895223\\
0.69123595505618	0.130361240526396\\
0.70123595505618	0.134371737770752\\
0.71123595505618	0.138532238654043\\
0.72123595505618	0.142854249540459\\
0.73123595505618	0.14735066721395\\
0.74123595505618	0.152036012683181\\
0.75123595505618	0.156926716505651\\
0.76123595505618	0.162041469932482\\
0.77123595505618	0.167401661051189\\
0.78123595505618	0.173031921984882\\
0.79123595505618	0.178960823075238\\
0.80123595505618	0.185221764384181\\
0.81123595505618	0.19185413630013\\
0.82123595505618	0.198904853679979\\
0.83123595505618	0.206430418875422\\
0.84123595505618	0.214499750647821\\
0.85123595505618	0.223198151144648\\
0.86123595505618	0.232633015196206\\
0.87123595505618	0.242942302062771\\
0.88123595505618	0.2543075736244\\
0.89123595505618	0.266974973504484\\
0.90123595505618	0.281290914688653\\
0.91123595505618	0.2977673129182\\
0.92123595505618	0.317213004811533\\
0.93123595505618	0.341038504351622\\
0.94123595505618	0.372143531150813\\
};
\addlegendentry{$n$ = 4, $q$ = 89}

\addplot [color=mycolor4]
  table[row sep=crcr]{%
0.0112359550561798	4.32665592115215e-17\\
0.0212359550561798	0.000299338432027253\\
0.0312359550561798	0.000917914617187265\\
0.0412359550561798	0.00169197822376203\\
0.0512359550561798	0.00256328957120354\\
0.0612359550561798	0.00350403916884092\\
0.0712359550561798	0.00449872094124841\\
0.0812359550561798	0.00553781495873544\\
0.0912359550561798	0.00661508544024959\\
0.10123595505618	0.007726258806281\\
0.11123595505618	0.00886831172855152\\
0.12123595505618	0.0100390589180887\\
0.13123595505618	0.0112369005599284\\
0.14123595505618	0.0124606603689542\\
0.15123595505618	0.0137094778163915\\
0.16123595505618	0.0149827341596161\\
0.17123595505618	0.0162800003467692\\
0.18123595505618	0.0176009995266173\\
0.19123595505618	0.0189455795797736\\
0.20123595505618	0.0203136926943171\\
0.21123595505618	0.0217053800020123\\
0.22123595505618	0.0231207599229259\\
0.23123595505618	0.0245600192782307\\
0.24123595505618	0.0260234065059342\\
0.25123595505618	0.0275112265015376\\
0.26123595505618	0.029023836735653\\
0.27123595505618	0.03056164439245\\
0.28123595505618	0.0321251043387469\\
0.29123595505618	0.0337147177816656\\
0.30123595505618	0.0353310315084323\\
0.31123595505618	0.0369746376286922\\
0.32123595505618	0.038646173760248\\
0.33123595505618	0.0403463236150751\\
0.34123595505618	0.0420758179551378\\
0.35123595505618	0.0438354358978147\\
0.36123595505618	0.0456260065592752\\
0.37123595505618	0.0474484110314696\\
0.38123595505618	0.0493035846948452\\
0.39123595505618	0.0511925198747994\\
0.40123595505618	0.0531162688554604\\
0.41123595505618	0.0550759472698362\\
0.42123595505618	0.0570727378908851\\
0.43123595505618	0.059107894853816\\
0.44123595505618	0.0611827483460114\\
0.45123595505618	0.063298709807635\\
0.46123595505618	0.0654572776933495\\
0.47123595505618	0.0676600438538172\\
0.48123595505618	0.0699087006050278\\
0.49123595505618	0.0722050485642322\\
0.50123595505618	0.0745510053436062\\
0.51123595505618	0.0769486152071508\\
0.52123595505618	0.0794000598130734\\
0.53123595505618	0.0819076701835815\\
0.54123595505618	0.084473940067192\\
0.55123595505618	0.0871015408861521\\
0.56123595505618	0.089793338494229\\
0.57123595505618	0.0925524120092117\\
0.58123595505618	0.0953820750313659\\
0.59123595505618	0.098285899615657\\
0.60123595505618	0.101267743434153\\
0.61123595505618	0.104331780648516\\
0.62123595505618	0.107482537114797\\
0.63123595505618	0.110724930668615\\
0.64123595505618	0.114064317394692\\
0.65123595505618	0.11750654497884\\
0.66123595505618	0.121058014483832\\
0.67123595505618	0.124725752197632\\
0.68123595505618	0.128517493592696\\
0.69123595505618	0.132441781934709\\
0.70123595505618	0.136508084724223\\
0.71123595505618	0.14072693199448\\
0.72123595505618	0.145110081592526\\
0.73123595505618	0.149670718035626\\
0.74123595505618	0.154423693500274\\
0.75123595505618	0.159385822168083\\
0.76123595505618	0.16457624281778\\
0.77123595505618	0.170016869657853\\
0.78123595505618	0.175732958612728\\
0.79123595505618	0.181753826649295\\
0.80123595505618	0.188113776910445\\
0.81123595505618	0.194853305082259\\
0.82123595505618	0.202020697009478\\
0.83123595505618	0.209674181706482\\
0.84123595505618	0.217884891068152\\
0.85123595505618	0.226741022536152\\
0.86123595505618	0.236353851245521\\
0.87123595505618	0.246866689658144\\
0.88123595505618	0.258468750764584\\
0.89123595505618	0.27141760824395\\
0.90123595505618	0.286077752031759\\
0.91123595505618	0.302991956306445\\
0.92123595505618	0.323027752195246\\
0.93123595505618	0.34772753022253\\
0.94123595505618	0.380390028647687\\
};
\addlegendentry{$n$ = 5, $q$ = 89}

\addplot [color=mycolor1, dashed]
  table[row sep=crcr]{%
0.000196811651249754	-1.95727534569762e-19\\
0.0101968116512498	0.00062371985347236\\
0.0201968116512498	0.00136053057323468\\
0.0301968116512498	0.00212725610586835\\
0.0401968116512498	0.0029137341221905\\
0.0501968116512498	0.00371613305328115\\
0.0601968116512498	0.00453258711165315\\
0.0701968116512498	0.00536207553260319\\
0.0801968116512498	0.00620401168937708\\
0.0901968116512498	0.00705805975465806\\
0.10019681165125	0.00792404173291248\\
0.11019681165125	0.00880188585378344\\
0.12019681165125	0.00969159593424321\\
0.13019681165125	0.0105932322323911\\
0.14019681165125	0.0115068990000838\\
0.15019681165125	0.0124327361427378\\
0.16019681165125	0.0133709135065564\\
0.17019681165125	0.0143216269093476\\
0.18019681165125	0.0152850953666933\\
0.19019681165125	0.01626155916234\\
0.20019681165125	0.0172512785317472\\
0.21019681165125	0.0182545328032307\\
0.22019681165125	0.0192716198899955\\
0.23019681165125	0.0203028560587624\\
0.24019681165125	0.0213485759227501\\
0.25019681165125	0.0224091326220886\\
0.26019681165125	0.023484898165621\\
0.27019681165125	0.0245762639159846\\
0.28019681165125	0.0256836412057208\\
0.29019681165125	0.0268074620766707\\
0.30019681165125	0.0279481801384308\\
0.31019681165125	0.0291062715444952\\
0.32019681165125	0.0302822360871107\\
0.33019681165125	0.0314765984139568\\
0.34019681165125	0.0326899093716572\\
0.35019681165125	0.0339227474829213\\
0.36019681165125	0.0351757205658616\\
0.37019681165125	0.0364494675058092\\
0.38019681165125	0.0377446601917969\\
0.39019681165125	0.0390620056318427\\
0.40019681165125	0.0404022482633058\\
0.41019681165125	0.041766172476941\\
0.42019681165125	0.0431546053758963\\
0.43019681165125	0.0445684197938475\\
0.44019681165125	0.0460085375998074\\
0.45019681165125	0.0474759333209522\\
0.46019681165125	0.0489716381191718\\
0.47019681165125	0.0504967441620704\\
0.48019681165125	0.0520524094349518\\
0.49019681165125	0.0536398630470463\\
0.50019681165125	0.0552604110930865\\
0.51019681165125	0.0569154431404873\\
0.52019681165125	0.0586064394231273\\
0.53019681165125	0.0603349788353727\\
0.54019681165125	0.0621027478349046\\
0.55019681165125	0.0639115503806081\\
0.56019681165125	0.0657633190528498\\
0.57019681165125	0.0676601275286231\\
0.58019681165125	0.0696042046142478\\
0.59019681165125	0.0715979500746864\\
0.60019681165125	0.073643952542591\\
0.61019681165125	0.0757450098437734\\
0.62019681165125	0.0779041521412746\\
0.63019681165125	0.0801246683806925\\
0.64019681165125	0.0824101366188641\\
0.65019681165125	0.0847644589415969\\
0.66019681165125	0.0871919018306898\\
0.67019681165125	0.0896971430350571\\
0.68019681165125	0.092285326247341\\
0.69019681165125	0.0949621252022725\\
0.70019681165125	0.0977338192182426\\
0.71019681165125	0.100607382729366\\
0.72019681165125	0.103590592043833\\
0.73019681165125	0.106692153474533\\
0.74019681165125	0.1099218582037\\
0.75019681165125	0.113290770885313\\
0.76019681165125	0.116811461233591\\
0.77019681165125	0.120498290954484\\
0.78019681165125	0.124367772743794\\
0.79019681165125	0.128439024305797\\
0.80019681165125	0.132734349388182\\
0.81019681165125	0.137279991201428\\
0.82019681165125	0.142107123784911\\
0.83019681165125	0.147253178098807\\
0.84019681165125	0.152763649159724\\
0.85019681165125	0.158694611525754\\
0.86019681165125	0.165116307416077\\
0.87019681165125	0.172118412708866\\
0.88019681165125	0.179818029833286\\
0.89019681165125	0.188372319557455\\
0.90019681165125	0.19799947562529\\
0.91019681165125	0.20901577951091\\
0.92019681165125	0.221906519935638\\
0.93019681165125	0.237477087289311\\
0.94019681165125	0.257226566610199\\
};
\addlegendentry{$n$ = 2, $q$ = 5081}

\addplot [color=mycolor2, dashed]
  table[row sep=crcr]{%
0.000196811651249754	-6.35065946987953e-18\\
0.0101968116512498	0.000693015736083385\\
0.0201968116512498	0.00147977495342732\\
0.0301968116512498	0.00228867017040475\\
0.0401968116512498	0.00311291366766564\\
0.0501968116512498	0.00395016382516505\\
0.0601968116512498	0.00479937727055277\\
0.0701968116512498	0.00566004726088465\\
0.0801968116512498	0.00653193550029055\\
0.0901968116512498	0.00741495619011763\\
0.10019681165125	0.00830911877613942\\
0.11019681165125	0.00921449691635147\\
0.12019681165125	0.0101312104719369\\
0.13019681165125	0.0110594145110673\\
0.14019681165125	0.0119992923396621\\
0.15019681165125	0.0129510509701577\\
0.16019681165125	0.0139149181342643\\
0.17019681165125	0.0148911403129949\\
0.18019681165125	0.0158799814615382\\
0.19019681165125	0.0168817222251723\\
0.20019681165125	0.017896659513939\\
0.21019681165125	0.0189251063483687\\
0.22019681165125	0.0199673919171584\\
0.23019681165125	0.0210238618065845\\
0.24019681165125	0.0220948783742103\\
0.25019681165125	0.0231808212483281\\
0.26019681165125	0.0242820879408987\\
0.27019681165125	0.0253990945663936\\
0.28019681165125	0.0265322766624437\\
0.29019681165125	0.0276820901109316\\
0.30019681165125	0.0288490121603985\\
0.31019681165125	0.0300335425525083\\
0.32019681165125	0.0312362047569845\\
0.33019681165125	0.0324575473209697\\
0.34019681165125	0.0336981453402521\\
0.35019681165125	0.0349586020612739\\
0.36019681165125	0.0362395506243851\\
0.37019681165125	0.0375416559604211\\
0.38019681165125	0.038865616854432\\
0.39019681165125	0.0402121681923258\\
0.40019681165125	0.0415820834083084\\
0.41019681165125	0.0429761771534066\\
0.42019681165125	0.0443953082080549\\
0.43019681165125	0.0458403826647952\\
0.44019681165125	0.0473123574106408\\
0.45019681165125	0.0488122439426678\\
0.46019681165125	0.050341112555015\\
0.47019681165125	0.0519000969408122\\
0.48019681165125	0.0534903992587349\\
0.49019681165125	0.0551132957210779\\
0.50019681165125	0.0567701427686268\\
0.51019681165125	0.0584623839074331\\
0.52019681165125	0.0601915572941281\\
0.53019681165125	0.0619593041700201\\
0.54019681165125	0.0637673782603017\\
0.55019681165125	0.0656176562738066\\
0.56019681165125	0.0675121496615245\\
0.57019681165125	0.069453017819347\\
0.58019681165125	0.071442582953281\\
0.59019681165125	0.0734833468649223\\
0.60019681165125	0.0755780099629518\\
0.61019681165125	0.077729492864888\\
0.62019681165125	0.0799409610249484\\
0.63019681165125	0.0822158529120744\\
0.64019681165125	0.0845579123714209\\
0.65019681165125	0.086971225938727\\
0.66019681165125	0.0894602660476642\\
0.67019681165125	0.0920299412857105\\
0.68019681165125	0.094685655128022\\
0.69019681165125	0.0974333749296978\\
0.70019681165125	0.100279713410032\\
0.71019681165125	0.103232025452755\\
0.72019681165125	0.106298523822552\\
0.73019681165125	0.109488418429097\\
0.74019681165125	0.112812085153772\\
0.75019681165125	0.116281272133646\\
0.76019681165125	0.119909353982111\\
0.77019681165125	0.123711648029309\\
0.78019681165125	0.127705811766144\\
0.79019681165125	0.131912348014324\\
0.80019681165125	0.136355255097523\\
0.81019681165125	0.141062875363993\\
0.82019681165125	0.14606901999093\\
0.83019681165125	0.151414486548755\\
0.84019681165125	0.157149148018035\\
0.85019681165125	0.16333489571942\\
0.86019681165125	0.17004989840679\\
0.87019681165125	0.177394965508367\\
0.88019681165125	0.18550342492752\\
0.89019681165125	0.194557194750659\\
0.90019681165125	0.204814532877304\\
0.91019681165125	0.216661828862095\\
0.92019681165125	0.23072123833514\\
0.93019681165125	0.248113388247466\\
0.94019681165125	0.271303663266969\\
};
\addlegendentry{$n$ = 3, $q$ = 5081}

\addplot [color=mycolor3, dashed]
  table[row sep=crcr]{%
0.000196811651249754	-5.79143794253702e-18\\
0.0101968116512498	0.000716821549647394\\
0.0201968116512498	0.0015178212460544\\
0.0301968116512498	0.0023379139753212\\
0.0401968116512498	0.00317173021508158\\
0.0501968116512498	0.00401751139157178\\
0.0601968116512498	0.00487452146093382\\
0.0701968116512498	0.0057424391513707\\
0.0801968116512498	0.00662114826613879\\
0.0901968116512498	0.00751064845954188\\
0.10019681165125	0.00841101178130036\\
0.11019681165125	0.00932235942000504\\
0.12019681165125	0.0102448483774148\\
0.13019681165125	0.011178663444547\\
0.14019681165125	0.0121240122004019\\
0.15019681165125	0.0130811218302241\\
0.16019681165125	0.0140502370914173\\
0.17019681165125	0.015031619034074\\
0.18019681165125	0.0160255442372112\\
0.19019681165125	0.0170323044107965\\
0.20019681165125	0.0180522062670541\\
0.21019681165125	0.0190855715976551\\
0.22019681165125	0.0201327375146167\\
0.23019681165125	0.0211940568266723\\
0.24019681165125	0.0222698985323313\\
0.25019681165125	0.0233606484174057\\
0.26019681165125	0.024466709749492\\
0.27019681165125	0.0255885040653769\\
0.28019681165125	0.0267264720499875\\
0.29019681165125	0.0278810745076371\\
0.30019681165125	0.0290527934280866\\
0.31019681165125	0.0302421331514666\\
0.32019681165125	0.0314496216375352\\
0.33019681165125	0.0326758118460695\\
0.34019681165125	0.0339212832365366\\
0.35019681165125	0.0351866433965544\\
0.36019681165125	0.036472529810101\\
0.37019681165125	0.0377796117779871\\
0.38019681165125	0.0391085925048053\\
0.39019681165125	0.0404602113684603\\
0.40019681165125	0.0418352463904953\\
0.41019681165125	0.0432345169277992\\
0.42019681165125	0.0446588866089809\\
0.43019681165125	0.0461092665417614\\
0.44019681165125	0.04758661882124\\
0.45019681165125	0.0490919603729315\\
0.46019681165125	0.0506263671690955\\
0.47019681165125	0.0521909788622643\\
0.48019681165125	0.0537870038860865\\
0.49019681165125	0.0554157250808492\\
0.50019681165125	0.0570785059094963\\
0.51019681165125	0.0587767973398555\\
0.52019681165125	0.0605121454804174\\
0.53019681165125	0.0622862000707277\\
0.54019681165125	0.0641007239436771\\
0.55019681165125	0.0659576035962442\\
0.56019681165125	0.0678588610282369\\
0.57019681165125	0.0698066670360722\\
0.58019681165125	0.0718033561817167\\
0.59019681165125	0.0738514436968339\\
0.60019681165125	0.0759536446306292\\
0.61019681165125	0.0781128956089287\\
0.62019681165125	0.0803323796443718\\
0.63019681165125	0.0826155545267236\\
0.64019681165125	0.0849661854327059\\
0.65019681165125	0.0873883825323515\\
0.66019681165125	0.0898866445414586\\
0.67019681165125	0.0924659093876407\\
0.68019681165125	0.0951316134346057\\
0.69019681165125	0.0978897610644557\\
0.70019681165125	0.100747006876637\\
0.71019681165125	0.103710753360144\\
0.72019681165125	0.106789267682165\\
0.73019681165125	0.109991822281429\\
0.74019681165125	0.113328865358125\\
0.75019681165125	0.116812229259486\\
0.76019681165125	0.120455387384652\\
0.77019681165125	0.12427377389421\\
0.78019681165125	0.128285185696262\\
0.79019681165125	0.132510293649459\\
0.80019681165125	0.136973300876772\\
0.81019681165125	0.141702802476833\\
0.82019681165125	0.146732926016878\\
0.83019681165125	0.152104871607214\\
0.84019681165125	0.157869034080394\\
0.85019681165125	0.164087996296442\\
0.86019681165125	0.170840867563861\\
0.87019681165125	0.178229777281659\\
0.88019681165125	0.186389978570632\\
0.89019681165125	0.195506337257594\\
0.90019681165125	0.205841919168737\\
0.91019681165125	0.217791659743729\\
0.92019681165125	0.231994916063191\\
0.93019681165125	0.249614544231223\\
0.94019681165125	0.273268044936257\\
};
\addlegendentry{$n$ = 4, $q$ = 5081}

\addplot [color=mycolor4, dashed]
  table[row sep=crcr]{%
0.000196811651249754	-1.59203378565175e-18\\
0.0101968116512498	0.000728594731444907\\
0.0201968116512498	0.00153595829507758\\
0.0301968116512498	0.00236087411189308\\
0.0401968116512498	0.00319871595630072\\
0.0501968116512498	0.0040480207396607\\
0.0601968116512498	0.00490820480228255\\
0.0701968116512498	0.00577903732668631\\
0.0801968116512498	0.00666046086536862\\
0.0901968116512498	0.00755251571599087\\
0.10019681165125	0.00845530339910445\\
0.11019681165125	0.0093689672682162\\
0.12019681165125	0.0102936814919409\\
0.13019681165125	0.0112296444830949\\
0.14019681165125	0.0121770748525629\\
0.15019681165125	0.0131362088782507\\
0.16019681165125	0.0141072989278009\\
0.17019681165125	0.0150906125081705\\
0.18019681165125	0.0160864317442211\\
0.19019681165125	0.0170950531627829\\
0.20019681165125	0.0181167877030788\\
0.21019681165125	0.0191519609019071\\
0.22019681165125	0.0202009132195633\\
0.23019681165125	0.0212640004840206\\
0.24019681165125	0.0223415944387263\\
0.25019681165125	0.0234340833848212\\
0.26019681165125	0.0245418729125016\\
0.27019681165125	0.0256653867191799\\
0.28019681165125	0.0268050675143495\\
0.29019681165125	0.0279613780128855\\
0.30019681165125	0.0291348020200708\\
0.31019681165125	0.0303258456130048\\
0.32019681165125	0.031535038424341\\
0.33019681165125	0.0327629350355505\\
0.34019681165125	0.0340101164881644\\
0.35019681165125	0.0352771919227756\\
0.36019681165125	0.0365648003569699\\
0.37019681165125	0.0378736126148884\\
0.38019681165125	0.039204333422798\\
0.39019681165125	0.0405577036869128\\
0.40019681165125	0.0419345029718054\\
0.41019681165125	0.0433355522001147\\
0.42019681165125	0.044761716596944\\
0.43019681165125	0.0462139089054015\\
0.44019681165125	0.047693092903255\\
0.45019681165125	0.049200287254691\\
0.46019681165125	0.0507365697358243\\
0.47019681165125	0.0523030818779707\\
0.48019681165125	0.0539010340789377\\
0.49019681165125	0.0555317112398301\\
0.50019681165125	0.0571964789933473\\
0.51019681165125	0.0588967905994647\\
0.52019681165125	0.0606341945960428\\
0.53019681165125	0.0624103433056591\\
0.54019681165125	0.0642270023162183\\
0.55019681165125	0.0660860610722163\\
0.56019681165125	0.0679895447365715\\
0.57019681165125	0.0699396275105063\\
0.58019681165125	0.0719386476321286\\
0.59019681165125	0.0739891243143874\\
0.60019681165125	0.0760937769316461\\
0.61019681165125	0.0782555468233309\\
0.62019681165125	0.0804776221556472\\
0.63019681165125	0.0827634663717327\\
0.64019681165125	0.0851168508713323\\
0.65019681165125	0.0875418926990797\\
0.66019681165125	0.0900430981935696\\
0.67019681165125	0.0926254137679833\\
0.68019681165125	0.0952942852710516\\
0.69019681165125	0.0980557277334317\\
0.70019681165125	0.100916407764949\\
0.71019681165125	0.103883741468163\\
0.72019681165125	0.106966011523071\\
0.73019681165125	0.110172508146646\\
0.74019681165125	0.113513700039773\\
0.75019681165125	0.117001443348803\\
0.76019681165125	0.120649239304125\\
0.77019681165125	0.124472554875289\\
0.78019681165125	0.128489225991555\\
0.79019681165125	0.132719970379712\\
0.80019681165125	0.137189048077346\\
0.81019681165125	0.141925124156626\\
0.82019681165125	0.146962413426738\\
0.83019681165125	0.152342226526454\\
0.84019681165125	0.158115100935781\\
0.85019681165125	0.164343807644891\\
0.86019681165125	0.17110771053159\\
0.87019681165125	0.178509294300784\\
0.88019681165125	0.186684327239555\\
0.89019681165125	0.19581845892492\\
0.90019681165125	0.206176024712147\\
0.91019681165125	0.218154199927852\\
0.92019681165125	0.232396815913825\\
0.93019681165125	0.250077640751174\\
0.94019681165125	0.273853560515398\\
};
\addlegendentry{$n$ = 5, $q$ = 5081}

\end{axis}
\end{tikzpicture}%

%% file: appendix.tex
\appendix
\subsection{Proof of \cref{lemma:sdiff0}} \label{app:optimal_delta}

The authors of \cite{xhemrishi2022efficient} numerically showed that for a sparse secret sharing scheme with two shares $\bfR$ and $\bfA + \bfR$ with sparsities $\sr$ and $\sar$, respectively, the minimal leakage is obtained for $\bfR$ and $\bfA+\bfR$ having the same desired sparsities $\sr=\sar$. \cref{lemma:sdiff0} proves this result theoretically. For completeness and self-containment of the paper, we describe here the setting consider in \cite{xhemrishi2022efficient}, which is as follows.

For desired sparsity levels of $\bfR$ and $\bfA+\bfR$, the authors of \cite{xhemrishi2022efficient} proposed a suitable PMF of $\rvR$ that minimizes the leakage. Thereby, the authors make the assumption of the entries $\rvA$ of $\bfA$ being independent among each other, hence also the entries of $\rvR$. Designing the PMF comes down to choosing $p_{r\vert a} \define \Pr(\rvR = r|\rvA = a)$ for all $r,a \in \mathbb{F}_q$ such that they form a proper PMF and meet the desired sparsity constraints for $\bfR$ and $\bfA + \bfR$. %
Intuitively, the choice of the $\pra$'s will affect the leakage through the shares. To minimize the leakage, we analyze the sum of the shares' leakages, which we ultimately find to be the same as minimizing the maximum.

While the authors of \cite{xhemrishi2022efficient} optimize the leakage for a given $\sr$ and $\sar$, we seek to optimize the leakage for a certain average sparsity $\savg\define\frac{\sr+\sar}{2}$. The difference of the parameters $\sar$ and $\sr$ is a degree of freedom that can be used to minimize the leakage. Hence, we also optimize over a parameter $\sdiff$ such that  $\sar = \savg + \sdiff$ and $\sr = \savg - \sdiff$. \cref{prop:general_optimization} states the general optimization problem required to solve to obtain the least possible leakage.
\begin{proposition} \label{prop:general_optimization}
Let $\mathcal{P} \define \{p_{r|a}: r,a\in \mathbb{F}_q\}$, then the minimum leakage for a given average sparsity $\savg$ is obtained by solving the following optimization problem:
\begin{align*}
    \lossopt &= \optimizerdelta \mutinf\left(\rvR; \rvA\right) + \mutinf\left(\rvA+\rvR; \rvA\right) \define \optimizerdelta \leakage{1} + \leakage{2} \\
    &=\optimizerdelta \begin{aligned}[t] &\kl{\paandr}{\pa \pr} + \kl{\paandapr}{\pa \papr} \end{aligned} \\
    &=\optimizerdelta \!\!\! \sum_{a,b\in\mathbf{F}_q} \!\! \pa(a) \!\left(\! \pra[b][a] \log\frac{\pra[b][a]}{\pr(b)} + \pra[b-a][a] \log\frac{\pra[b-a][a]}{\papr(b)} \!\right)\!
\end{align*}
subject to %
\begin{align*}
    \forall a\in\Fq: \pra[0][a] + \sum_{r\in\Fqstar} \pra[r][a] - 1 &= 0, \\[-5pt]
    \pra[0][0] \cdot \pa(0) + \sum_{a\in\Fqstar} \pra[0][a] \cdot \pa(a) - \savg + \sdiff &= 0, \\[-5pt]
    \pra[0][0] \cdot \pa(0) + \sum_{a\in\Fqstar} \pra[-a][a] \cdot \pa(a) -\savg - \sdiff &= 0.
\end{align*}
\end{proposition}

Based on the assumption of entries in $\rvA$ being identically distribution among the non-zero entries, we investigate the same scheme as proposed in \cite{xhemrishi2022efficient}, which we give for completeness in the following:
\begin{align}
 \label{eq:dependent_on_0}
    \Pr\{\Rij = r \lvert \Aij = 0\} \!&\!= \begin{cases} 
    \pz, &r = 0 \\   \pzinv , &r \neq 0,
    \end{cases} %
\end{align}
\begin{align}
\label{eq:dependent_on_nz}
    \Pr\{\Rij = r \lvert \Aij = a\} \!&\!= \begin{cases} 
    \pext, &r = 0 \\ \pnz, &r = -a, \\
    p_{2,3}^{\text{inv}} , &r \not\in \{0,-a\}, \\
    \end{cases}
\end{align}
where $r\in \F_q$, $a \in \F_q^*$, and $-a$ is the additive inverse of $a$ in $\F_q^*$. Further, $\pz,\pext,\pnz, \pzinv \define (1-\pz)/(q-1)$ and $p_{2,3}^{\text{inv}} \define (1-\pext-\pnz)/(q-2)$ are non-negative and at most $1$. Numerical results have shown \cite{xhemrishi2022efficient} that given the entries of $\bfA$ being i.i.d. and the non-zero entries uniformly distributed over $\Fq^*$, optimizing over PMFs as in \eqref{eq:dependent_on_0} and \eqref{eq:dependent_on_nz} is equivalent to optimizing over $q^2$ unknowns $\pra$. This justifies the usage of PMFs as in \eqref{eq:dependent_on_0} and \eqref{eq:dependent_on_nz}.

As a result of the construction above, the sparsities of $\sar$ and $\sr$ are given by \cref{lemma:sparsity_delta}.
\begin{lemma}\label{lemma:sparsity_delta}
We follow the rules in \eqref{eq:dependent_on_0} and \eqref{eq:dependent_on_nz} to construct two matrices $\bfR$ and $\bfA+\bfR$ dependently on $\bfA$ with sparsity $s$. Without loss of generality, we assume that the difference in sparsities between $\sar$ and $\sr$ is $2\sdiff > 0$, since by construction the scheme is fully symmetric in $\sr$ and $\sar$. The sparsity levels result in
\begin{align*}
    \sr &\define \savg - \sdiff = \pz s + \pext (1-s), \\
    \sar &\define \savg + \sdiff = \pz s + \pnz(1-s).
\end{align*}
\end{lemma}
\begin{proof}
    This result directly follows from the construction in \eqref{eq:dependent_on_0} and \eqref{eq:dependent_on_nz} and the PMF of $\rvA$.
\end{proof}

Using the conditional PMF from above, we solve the optimization problem stated in \cref{prop:general_optimization}, which can be reformulated as given in \cref{lemma:leakage_delta}.

\begin{lemma} \label{lemma:leakage_delta}
Considering PMFs of the form given in \eqref{eq:dependent_on_0} and \eqref{eq:dependent_on_nz}, the total leakage $\totalleakagesdiff \define \mathrm{L}_1+\mathrm{L}_2$ is given by
\begin{align*}
&\totalleakagesdiff = s \bigg[ \xlogx{\pz}{\sar} + \xlogx{\pz}{\sr} \\
&+(q-1) \xlogx{\pzinv}{\sarinv} + \xlogx{\pzinv}{\srinv} \bigg] \nonumber\\
& \!+\!(1-s) \!\cdot\! \bigg[ \xlogx{\pext}{\sarinv} + \xlogx{\pext}{\sr} + \xlogx{\pnz}{\sar} \\ &+ \xlogx{\pnz}{\srinv} +(q-2) \xlogx{\ptwoinv}{\sarinv} + \xlogx{\ptwoinv}{\srinv} \bigg],
\end{align*}
subject to
\begin{align}
    c_1(\pz,\pzinv) &\define \pz + (q-1) \pzinv - 1 = 0, \label{eq:constraint1} \\
    c_2(\pext,\pnz,\pnzinv) &\define \pext + \pnz + (q-2) \pnzinv - 1 = 0, \label{eq:constraint2} \\
    c_3(\pz,\pext) &\define \pz s + \pext (1-s) - \sr = 0, \label{eq:constraint3} \\
    c_4(\pz,\pnz) &\define \pz s + \pnz (1-s) - \sar = 0. \label{eq:constraint4}
\end{align}
where $\sr \define \savg - \sdiff$, $\sar = \savg + \sdiff$, $\srinv \define (1-\savg+\sdiff)/(q-1)$, $\sarinv \define (1-\savg-\sdiff)/(q-1)$.
\begin{proof}
We obtain this result by expressing the mutual information in terms of KL-divergences, using the conditional of $\rvR$ given $\rvA$ in \eqref{eq:dependent_on_0} and \eqref{eq:dependent_on_nz}, and simplifying.
\end{proof}
\end{lemma}

We are now ready to prove \cref{lemma:sdiff0}.

\begin{proof}[Proof of \cref{lemma:sdiff0}]
To determine the optimal $\sdiff$ that minimizes the leakage, we utilize the method of Lagrange multipliers to combine the objective function, i.e., the total leakage from \cref{lemma:leakage_delta}, with the constraints in \eqref{eq:constraint1}-\eqref{eq:constraint4}. Thus, the objective function to be minimized can be expressed as
\begin{align*}
    \lossabbrev &\define \totalleakagesdiff + \lambda_1 c_1(\pz,\pzinv) + \\
    & + \lambda_2 c_2(\pext,\pnz,\pnzinv) + \lambda_3 c_3(\pz,\pext) + \lambda_4 c_4(\pz,\pnz).
\end{align*}

This objective can be minimized by setting the gradient to zero, i.e., $\grad \lossabbrev = 0$, which amounts to solving a system of eleven equations with eleven unknowns. %

This is a convex optimization problem, since being a sum of convex functions in each of the unknowns and the constraints being affine in the unknowns. The function $\totalleakagesdiff$ is convex in $\sdiff$ by the composition theorem \cite[Ch. 3.2.4]{boyd2004convex} since $f(x) = -\log(x)$ is convex and $g(x) = \savg \pm \sdiff$ is convex. Hence, $f(g(x))$ is convex. For the variables $\pz,\pzinv,\pext,\pnz,\pnzinv$, one can show convexity by the perspective of the function $f(x) = -\log(x)$, which is $g(x,t) = t f(x/t) = t \log(t/x)$ \cite[Ex. 3.19]{boyd2004convex} and preserves convexity \cite[Ch. 3.2.6]{boyd2004convex}. Equivalently, one can argue by the convexity of KL-divergence.

We now state the system of equations given by $\nabla_{\pext,\pnz,\sdiff} \lossabbrev = 0$:
\begin{align}
    \frac{\partial \lossabbrev}{\partial \sdiff} &= s \bigg[ \pz \left(-\frac{1}{\sar} + \frac{1}{\sr}\right) + \nonumber \\
    & + (q-1)\pzinv \left( \frac{1}{1\!-\!\sar} - \frac{1}{1\!-\!\sr}\right)\bigg] + \nonumber \\
    & + (1-s) \! \bigg[ p_2 \! \left( \frac{1}{1\!-\!\sar} \! + \! \frac{1}{\sr} \right) \! - \! p_3 \left( \frac{1}{\sar} \! + \! \frac{1}{1\!-\!\sr} \right) \nonumber \\
    & + (q-2) \ptwoinv \! \left( \frac{1}{1\!-\!\sar} - \frac{1}{1\!-\!\sr}\right) \bigg] \! + \! \lambda_3 \! - \! \lambda_4 = 0 \label{eqline:sys1} \\
    \frac{\partial \lossabbrev}{\partial p_2} &= (1-s)\left(2\log(p_2) - \log(\sr) - \log(\sarinv) + 2\right)  \nonumber \\
    &+ \lambda_2 + (1-s) \lambda_3 = 0 \label{eqline:partial_p2} \\
    \frac{\partial \lossabbrev}{\partial p_3} &= (1-s)\left(2\log(p_3) - \log(\srinv) - \log(\sar) + 2\right) \nonumber \\
    &+ \lambda_2 + (1-s) \lambda_4 = 0 \label{eqline:partial_p3}
\end{align}
We will first calculate $\lambda_3 - \lambda_4$ from \eqref{eqline:partial_p2} and \eqref{eqline:partial_p3} to be later substituted in \eqref{eqline:sys1}. By computing $\eqref{eqline:partial_p2} - \eqref{eqline:partial_p3}$ and dividing the result by $(1-s)$, we obtain
\begin{align}
    &\lambda_3-\lambda_4 = \nonumber \\ 
    &= 2 (\log(p_3) - \log(p_2)) + \log\left(\frac{\sarinv}{\sar}\right) + \log\left(\frac{\sr}{\srinv}\right) \! \label{eqline:lambda_diff1} \\
    &= 2\log\left(\frac{\savg+\sdiff-s p_1}{1-s}\right) - 2\log\left(\frac{\savg-\sdiff-s p_1}{1-s}\right) + \nonumber \\
    &+ \log\left(\frac{1-\savg-\sdiff}{(q-1)(\savg+\sdiff)}\right) + \log\left(\frac{(q-1)(\savg-\sdiff)}{1-\savg+\sdiff}\right) \label{eqline:lambda_diff2} \\ 
    &= 2 \log\! \left(\! \frac{\savg+\sdiff-s p_1}{\savg-\sdiff-s p_1} \! \right) \!\! + \! \log\!\left(\! \frac{(1-\savg-\sdiff)(\savg-\sdiff)}{(1-\savg+\sdiff)(\savg+\sdiff)} \! \right) %
    \label{eqline:lambda34}
\end{align}
where from \eqref{eqline:lambda_diff1} to \eqref{eqline:lambda_diff2} we used \eqref{eq:constraint3} and \eqref{eq:constraint4} to express $p_2$ and $p_3$ in terms of $p_1$, respectively.
Simplifying \eqref{eqline:sys1} by use of $\nabla_{\lambda_1\lambda_2,\lambda_3,\lambda_4,\lambda_5} \lossabbrev = 0$, i.e., the constraints in \eqref{eq:constraint1}-\eqref{eq:constraint4} yields
\begin{align}
    &\frac{\partial \lossabbrev}{\partial \sdiff} - (\lambda_3 - \lambda_4) \nonumber \\
    & = \frac{1}{\sr} (s p_1 + (1-s) p_2) - \frac{1}{\sar} (s p_1 + (1-s)p_3)  + \nonumber \\
    & + \frac{1}{1-\sar} \underbrace{(s(q-1) \pzinv + (1-s) p_2 \! + \! (1-s)(q-2)\ptwoinv)}_{(a)} \nonumber \\
    & - \frac{1}{1-\sr} \underbrace{((q-1)\pzinv + (1-s)p_3 + (q-2)(1-s)\ptwoinv)}_{(b)} \nonumber \\
    &= \frac{\sr}{\sr} -\frac{\sar}{\sar} + \frac{1-\sar}{1-\sar} - \frac{1-\sr}{1-\sr} = 0, \label{eqline:partial_sdiff_simplified}
\end{align}
where we used that $\sr = s p_1 + (1-s) p_2$ from \eqref{eq:constraint3}, $\sar = s p_1 + (1-s) p_3$ from \eqref{eq:constraint4}. Further, we obtain $1-\sar = (a)$ by solving \eqref{eq:constraint1} for $p_1$, substituting $p_1$ in \eqref{eq:constraint4} and solving for $p_3$ and finally inserting in \eqref{eq:constraint2}, multiplying by $(1-s)$ and solving for $1-\sar$. Lastly, we obtain $1-\sr = (b)$ by solving \eqref{eq:constraint1} for $p_1$, substituting $p_1$ in \eqref{eq:constraint3} and solving for $p_2$ and finally inserting in \eqref{eq:constraint2}, multiplying by $(1-s)$ and solving for $1-\sr$.
Putting \eqref{eqline:lambda34} and \eqref{eqline:partial_sdiff_simplified} together, we obtain
\begin{align*}
    \frac{\partial \lossabbrev}{\partial \sdiff} + \lambda_3 - \lambda_4 = 2& \log\left(\frac{\savg+\sdiff-s p_1}{\savg-\sdiff-s p_1} \right) + \\
    &+ \log\left(\frac{(1-\savg-\sdiff)(\savg-\sdiff)}{(1-\savg+\sdiff)(\savg+\sdiff)} \right) = 0. 
\end{align*}
Since the optimization is convex and hence there only exists a unique solution, it follows that $\sdiff=0$ minimizes the objective in \cref{prop:general_optimization}. Since choosing $\sdiff \neq 0$ would increase the leakage through one of the shares, this is equivalent to minimizing the maximum of the shares' leakages.
\end{proof}